\documentclass[preprint,12pt,authoryear]{elsarticle}

\usepackage[margin=0.9in]{geometry}
\usepackage[utf8]{inputenc}
\usepackage[english]{babel}
\usepackage{amsmath,amssymb,amsthm}
\allowdisplaybreaks 
\usepackage{graphicx}
\usepackage{float}
\usepackage{makecell}
\usepackage{algorithm}
\usepackage{algpseudocode}
\usepackage{booktabs}
\usepackage{xltabular}
\usepackage{xcolor}
\usepackage{multirow}
\usepackage{longtable}
\usepackage{tabularx}
\usepackage[hidelinks]{hyperref}
\usepackage{caption}

\biboptions{authoryear}

\newtheorem{proposition}{Proposition}
\theoremstyle{remark}
\newtheorem{remark}{Remark}

\journal{Preprint}

\begin{document}

\begin{frontmatter}

\title{Three-sided mobility-energy market design as a multiperiod stochastic assignment game}

\author[nyu]{Hai Yang}
\author[nyu]{Joseph J. Y. Chow\corref{cor1}}
\cortext[cor1]{Corresponding author}
\ead{joseph.chow@nyu.edu}

\address[nyu]{C2SMART Center, Department of Civil, Urban, and Environmental Engineering, New York University Tandon School of Engineering, 6 MetroTech Center, Brooklyn, NY 11201, USA}

\begin{abstract}
As mobility service providers (MSPs) and energy providers (EPs) expand electric vehicle ecosystems, models are needed to understand their interactions within a three-sided market. Existing frameworks often overlook the temporal interdependencies between mobility and charging demands. We address this gap by proposing a bilevel problem as an assignment game overseen by a market regulator. The upper level optimizes service pricing to maximize platform profitability. The lower level models a multi-stakeholder equilibrium using a scalable, link-based Perturbed Utility Route Choice (PURC) framework. The evaluation time frame is divided into discrete intervals, capturing the temporal lag between mobility and charging demand via an empirical affine function. We solve the model by chronologically decomposing the lower level into interacting mobility service and recharge subnetworks. Numerical experiments on the expanded Nguyen-Dupuis network reveal several key insights. First, a critical charging capacity threshold exists; operating below it forces a severe reduction in the deployable fleet and creates localized transit deserts. Second, modeling endogenous operating costs reveals a concave profit trajectory, demonstrating that total profit maximizes at a specific fleet size just before market saturation. Third, optimal dynamic pricing operates within a narrow range, where peak pricing acts as a steady revenue driver and off-peak pricing serves as a highly sensitive operational buffer. These findings provide actionable strategies for coordinating fleet sizing and charging infrastructure deployment.
\end{abstract}

\end{frontmatter}

\section{Introduction}\label{sec:Intro}

Sustainable cities must reduce non-renewable energy use in transportation, which accounts for 25\% of global output \citep{conti2016international}. Yet electrification alone is insufficient. Prior studies show that uncoordinated charging can substantially increase peak electricity demand, while the environmental benefits of electric vehicles remain sensitive to the energy mix and operating conditions \citep{van2011energy, ma2012evaluation, huo2015life}. A key sustainability strategy is therefore to shift travel toward more efficient, renewable-oriented modes, including shared and pooled mobility services and micromobility. However, these services remain difficult to operate sustainably because long charging times and limited charging infrastructure raise operating costs \citep{jung2014stochastic, jung2019effects}, and their performance varies across built environments, creating challenges such as the first- and last-mile problem \citep{chang1991multiple}.

In a mobility--energy ecosystem, charging demand at one time of day is driven by mobility demand at other times, while the cost and capacity of energy provision affect mobility fares, traveler choices, and, in turn, energy-provider profitability. This interdependence fundamentally changes the mathematical structure of the market: once energy providers are treated as strategic actors, the conventional two-sided market of travelers and mobility providers becomes a three-sided market. Modeling and governing such a system therefore requires a framework that jointly captures traveler routing, mobility service assignment, and physical energy constraints.

Existing literature does not yet provide such an integrated approach. Instead, two related but largely parallel streams have developed. One focuses on EV fleet planning, routing, and charging infrastructure, showing that electrification introduces additional spatiotemporal constraints, longer service interruptions, and stronger dependence on charging-facility availability \citep{li2024toward}. Studies on electric routing and fleet operations show that charging behavior shapes feasible routing and service decisions \citep{montoya2017electric}, and that fleet operations, battery depletion, and charging choices must be considered jointly in electrified systems \citep{jung2014stochastic, jung2019effects, pantelidis2022node, yi2021framework, kullman2022dynamic, zhan2022simulation}. Planning studies likewise integrate transportation and power considerations into charging infrastructure and fleet design \citep{nourinejad2016equilibrium, wang2018coordinated, zhang2018optimal, chen2020optimal}. More recent work shows that charging stations and battery swapping may serve as complementary facilities for electrified ride-hailing fleets \citep{lai2024multimodalcharging}, and that fleet management, pricing, and energy costs are jointly shaped in coupled transport--power planning settings \citep{ding2022pricingfleet}. Together, these studies show that charging is a core system-design component rather than a secondary operational detail.

A second stream examines mobility market design, platform pricing, operator competition, and traveler participation. MaaS and MOD studies use assignment-game and bilevel formulations to represent strategic interactions among users, operators, and platforms \citep{daou2024modelling, rasulkhani2019route, pantelidis2020many, liu2024demand, yao2024design, xi2024single, bandiera2024mobility}. 

An assignment game, unlike a noncooperative game, is an assignment problem that also determines the cost transfers between players to sufficiently incentivize them to participate in the coalition. In a one-to-one stable matching assignment game without network effects, the problem can be represented as a simple linear program where the dual variables can be used to determine allocations \citep{shapley1971assignment}. \cite{liu2024demand} showed that for multimodal networks, such assignment games can be represented in a bilevel framework, which can further be extended to a unique set of cost transfers when the lower level matches are stochastically assigned with endogenous capacities \citep{liu2024modeling}. \citet{yang2025bilevel} show that a bilevel mobility-hub design model with a lower-level link-based perturbed utility route choice (PURC) assignment can achieve strong computational scalability. These studies provide an important methodological foundation by endogenizing pricing, capacity allocation, cost transfers, and traveler participation within a platform-based equilibrium. However, they remain largely two-sided, as charging-related decisions and the strategic role of energy providers are omitted or only indirectly represented.

Recent electrified mobility studies move closer to integration. Data-driven studies identify the spatiotemporal distribution of charging demand from travel demand or ride-hailing activity \citep{zhang2020shareduse, li2022trajectory, jin2025chargingdemand}. Other work shows that charging cost in electric ride-hailing depends not only on electricity price but also on waiting time, and that these charging costs directly affect the equilibrium between service supply and demand \citep{chen2024modeling}. Studies of electric autonomous MOD systems similarly emphasize that charging demand is linked to passenger demand and shaped by both platform and government decisions \citep{gao2024pluginswap}. On the infrastructure and operations side, the energy sector is increasingly modeled as an active decision maker \citep{ding2022pricingfleet, cai2024optimizing, zhao2025charging}. More broadly, \citet{garcia2019state} conceptualize eMaaS as an ecosystem of interdependent participants, while \citet{xi2024strategizing} examine how EV-related incentives affect electric mobility use under competition.

Despite this progress, the literature still lacks a unified network-based equilibrium framework that jointly models mobility assignment and recharge assignment over multiple periods with explicit strategic interactions among travelers, mobility providers, and energy providers. We address this gap by proposing a three-sided assignment game with a bilevel framework for a mobility--energy market with explicit time-of-day interdependencies between mobility demand and charging demand. The upper level represents a market regulator, interpreted as either a public agency or a private platform hosting the three groups of market participants, and determines cost-transfer decisions to maximize market profit distributed between mobility and energy providers. The lower level captures the equilibrium decisions of travelers, service operators, and energy providers across multiple periods. To represent the temporal link between mobility demand and charging demand, we propose using an affine function to relate them across periods as shown to be effective by \citet{jin2025chargingdemand}. Although the present model focuses on fixed charging stations as a starting point, it provides a foundation for future extensions to other charging technologies like mobile charging and battery swapping.

The problem is computationally challenging because, relative to earlier bilevel mobility design models such as \citet{yang2025bilevel}, it introduces time-of-day coupling, energy-provider decision variables, and fleet-positioning decisions associated with recharging. To address this complexity, we retain the scalable link-based PURC formulation of \citet{fosgerau2022perturbed}, as adopted in \citet{yang2025bilevel}, and further decompose the lower-level equilibrium into two interconnected subnetworks: a mobility service network for passenger movement and a recharge service network for charging and fleet rebalancing. This decomposition preserves the key strategic and physical interactions while keeping the model tractable for planning analysis.

We make three contributions. First, we develop a network-based bilevel framework as an electrified mobility market assignment game that explicitly integrates mobility assignment and recharge assignment over multiple time periods. Second, we extend the predominantly two-sided mobility market-design literature by endogenizing energy providers within a three-sided market. Third, we propose a scalable decomposition-based solution method and evaluate it on an expanded Nguyen--Dupuis network to analyze both lower-level equilibrium behavior and upper-level profit optimization.

\section{Proposed model and solution method}

We propose a strategic planning model assuming perfect information symmetry between the regulator and the operators, not mechanism design. As such, the detailed mechanisms to ensure truth-telling (i.e. incentives compatibility) are outside the scope of this study; however, we do need to consider individual rationality (IR) of travelers and operators.

\subsection{Network structure and model assumptions}

We consider a network served by one or more MOD services $m \in M$ with exclusive access to one or more energy providers $e \in E$, regulated by a platform $\mathcal{P}$. The regulator can be a public agency or a private enterprise, or even the mobility service if it is the sole mobility service in the platform. Other mobility services and modes are considered to be out-of-platform options. The services operate in a cyclic operating horizon, such as a 24-hour cycle, which is partitioned into a set of discrete time intervals $\mathcal{T}$. Each interval $t \in \mathcal{T}$ exhibits a distinct, steady-state travel demand  $\mathbf{q}_t$ that includes users of both the platform mobility services $\mathbf{q}^I_t$ and out-of-platform options $\mathbf{q}^O_t$, where $\mathbf{q}^O_t + \mathbf{q}^I_t = \mathbf{q}_t$. 

\subsubsection{Cyclic temporal horizon}
The mobility--energy platform is represented as a multilayer network composed of two interconnected components: mobility service subnetworks $G$ and recharge service subnetworks $\tilde{G}$, as shown in \textbf{Fig.~\ref{fig:eMaaS network}}. This structure captures the spatial and temporal coupling between EV-based mobility service and charging service. It is not intended to represent dynamic traffic effects such as spillovers or FIFO conditions; rather, it provides a macroscopic steady-state representation of time-of-day interactions.

The mobility service subnetworks, shown as blue layers in \textbf{Fig.~\ref{fig:eMaaS network}}, represent the spatial assignment of mobility flows in each time interval under steady-state conditions. Although each layer has its own demand pattern and cost structure, together they represent the variation in urban mobility over the day. To model a cyclic steady-state day, the demand pattern in the first interval, $\mathbf{q}_1$, is assumed to follow that of the final interval, $T$. This assumption allows the model to capture overnight capacity carryover and long-run system consistency.

Between consecutive mobility layers are recharge service subnetworks, shown as green layers in \textbf{Fig.~\ref{fig:eMaaS network}}. These subnetworks serve two purposes: they assign charging demand spatially to designated stations and redistribute fleet capacity as vehicles move between active service and charging states. In this way, the network captures the temporal trade-off between service availability and vehicle state-of-charge. The arc definitions and mathematical properties of both subnetworks are introduced below.

\begin{figure}[h]
    \centering
    \includegraphics[width=0.8\textwidth]{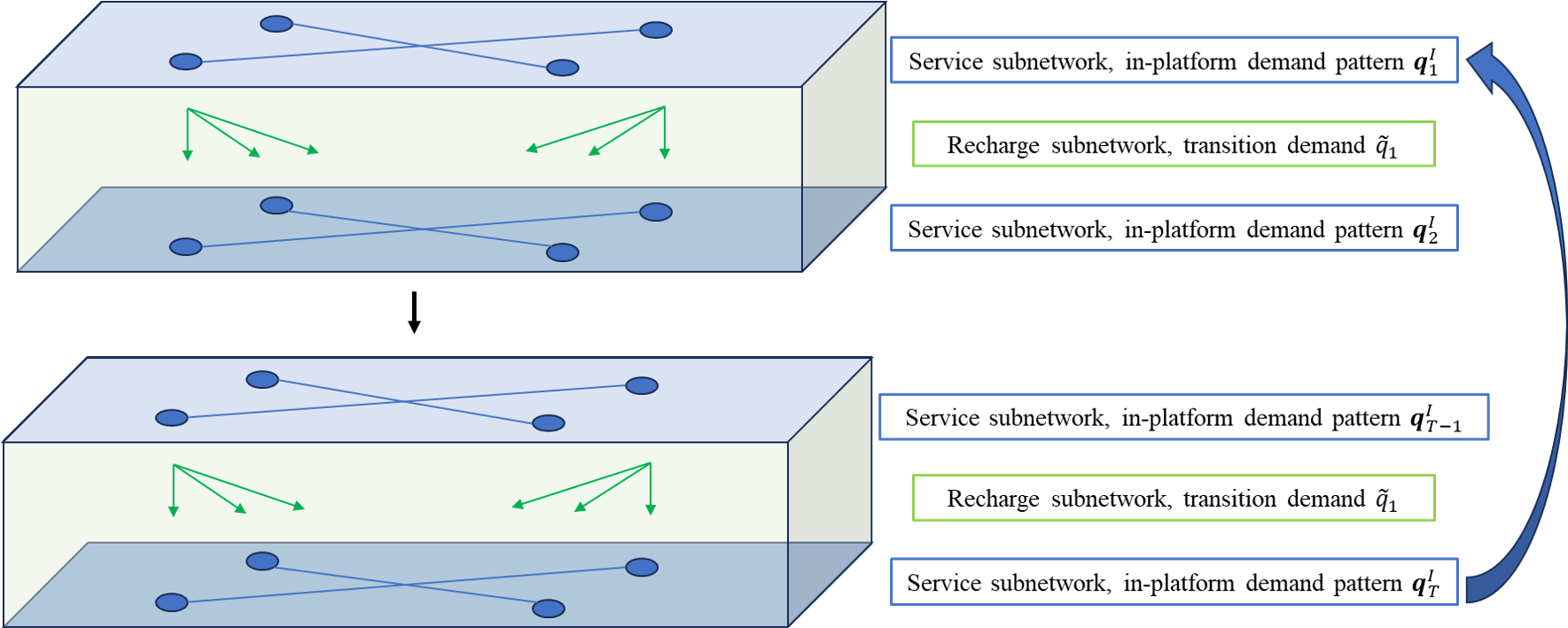}
    \caption{Illustration of mobility-energy platform network structure}
    \label{fig:eMaaS network}
\end{figure}

\subsubsection{Mobility service subnetwork}

We first describe the mobility service subnetwork $G$. Although each time layer $t \in \mathcal{T}$ has different node capacities and demand levels, the underlying topology is identical across all intervals. \textbf{Fig.~\ref{fig:service network}} illustrates the detailed structure of one such layer within the broader framework of \textbf{Fig.~\ref{fig:eMaaS network}}.

The service network is denoted by $G=(N,A)$, where $N$ is the set of service-related nodes and $A$ is the set of service links used to satisfy origin--destination (OD) demand. Demand is aggregated at centroid nodes $\{O \cup D\} \subseteq N$. The network supports two categories of service: the platform's on-demand MOD services and a composite out-of-platform (OOP) option. The OOP option is represented by the lower subgraph in \textbf{Fig.~\ref{fig:service network}}, consisting of links $A_f$ (brown links) between centroids.

In parallel, each MOD fleet $m \in M$ operates its own subnetwork $G_m$, represented by the upper subgraph in \textbf{Fig.~\ref{fig:service network}}. This subgraph contains mobility service nodes $N_m$ (blue nodes) and service links $A_m$ (solid blue links). These links describe average travel performance between nodes rather than real-time routing outcomes, consistent with macroscopic ride-pooling network models \citep{yang1998network, xu2021equilibrium, yang2025bilevel}. Travelers move between the physical demand subgraph and each MOD subgraph through access and egress links, shown as dashed blue arrows. The access links represent the average waiting cost through a "meeting function" \citep{yang2010equilibria, liu2024demand}. When multiple MOD operators compete, each is represented by a separate service subgraph connected to the common centroids through its own access and egress links.

To represent operational constraints such as cruising and congestion, MOD access links are modeled as store-and-forward links with service queues. Congestion effects, including waiting delay for vehicle pickup, are governed by the capacity of the access links. Accordingly, each MOD operator $m$ chooses the service capacity of its access links $l \in A_m^+$ in each interval $t$, thereby limiting inflow to the service network and capturing queueing delays caused by supply--demand imbalance.

\begin{figure}[h]
    \centering
    \includegraphics[width=0.7\textwidth]{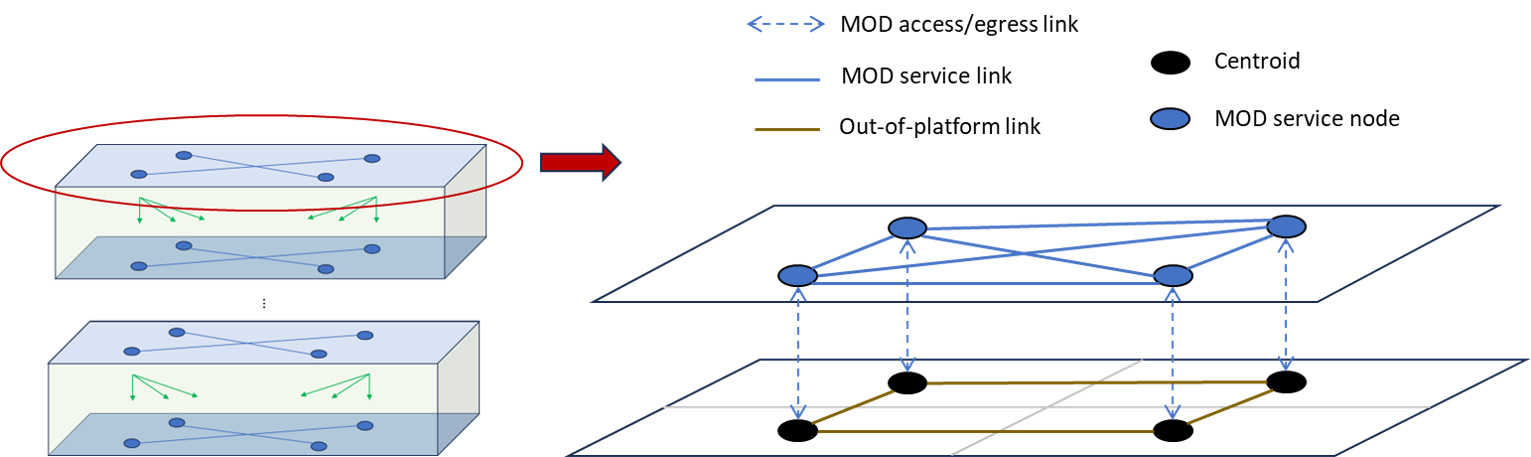}
    \caption{Illustration of mobility service subnetwork structure}
    \label{fig:service network}
\end{figure}

\subsubsection{Charging service subnetwork}

As shown in \textbf{Fig.~\ref{fig:eMaaS network}}, the recharge service subnetwork $\tilde{G}$ forms the temporal interface between two consecutive mobility service layers. Charging is modeled as a discrete transition phase linking operations in interval $t$ to those in interval $t+1$. This construction discretizes continuous time while preserving the flow conservation needed for multiperiod fleet management.

The main purpose of this subnetwork is to model redistribution flow: the movement of average vehicle presence needed to satisfy different steady-state capacity requirements across time intervals, together with additional idle vehicles caused by charging or demand--supply imbalance. This differs from \cite{yang2025bilevel}, which assumes a time-invariant steady-state fleet distribution over service nodes, and from rebalancing models that track individual vehicle locations \citep{nourinejad2015vehicle, sayarshad2017non}. Instead, the present framework assumes that an operator's steady-state fleet distribution changes over time and must be redistributed accordingly to preserve flow capacity. The charging subnetwork therefore defines OD pairs connecting every mobility service node in layer $t$ to every mobility service node in layer $t+1$ for each operator $m \in M$. These flows ensure that the average fleet capacity assigned to each MOD node $i \in N_m$ is properly transferred and conserved between consecutive intervals.

Within this structure, redistribution can occur along two path types, illustrated in \textbf{Fig.~\ref{fig:charging network}}. Vehicles that require charging follow charging-based links (solid green), passing through a charging node $i \in \tilde{N}_E \subseteq \tilde{N}$ before reaching their destination in layer $t+1$. Vehicles that do not need charging use fleet redistribution links (dashed purple arrows) to move directly between service nodes across intervals. This topology includes both direct cross-period links and indirect paths through charging infrastructure, thereby coupling spatial relocation with energy management. Redistribution patterns can be calibrated through link weights. For example, nighttime charging policies or redistribution limited to nearby zones can be represented by assigning high weights to daytime charging links or to transfers toward distant service nodes.

A complete list of sets, parameters, and decision variables is provided in as follows. Key notation is also introduced inline when first used in the model.

\begin{itemize}
    \item $\mathcal{T}$: Set of discrete time intervals representing a cyclic planning horizon.
    \item $\mathcal{M}$: Set of Mobility-on-Demand (MOD) fleet operators.
    \item $E$: Set of energy providers/charging operators.
    \item $\mathcal{N}$: Set of service network nodes.
    \item $\mathcal{\tilde{N}}$: Set of charging network nodes.
    \item $\mathcal{N}_m$: Subset of service nodes managed by operator $m$.
    \item $\mathcal{\tilde{N}}_e$: Subset of charging station nodes managed by energy provider $e$.
    \item $\mathcal{A}$: Set of service network links.
    \item $\tilde{\mathcal{A}}$: Set of recharge network links.
    \item $\mathcal{A}_m^a$: Subset of MOD access links.
    \item $\tilde{\mathcal{A}}_i^a$: Subset of links entering charging node $i$.
    \item $\mathcal{S}$: Set of Origin-Destination (OD) pairs for travelers.
    \item $s^{+}, s^{-}$: Origin and destination of OD pair $s\in S$.
    \item $\tilde{\mathcal{S}}$: Set of OD pairs for fleet redistribution in the recharge network.
    \item $\tilde{s}^{+}, \tilde{s}^{-}$: Origin and destination of redistribution OD pair $\tilde{s} \in \tilde{\mathcal{S}}$.
\end{itemize}

\subsection*{Parameters and Input Variables}
\begin{itemize}
    \item $q_s^t$: Travel demand for OD pair $s \in \mathcal{S}$ in time interval $t \in \mathcal{T}$.
    \item $\pi_{t}^{\tau}$: Propagation coefficient of the autoregression function connecting service demand in interval $\tau$ to interval $t$
    \item $v_l$: Maximum physical capacity scaling factor for access link $l$.
    \item $h_i$: Maximum physical capacity scaling factor for node $i$.
    \item $d_l$: Length of link $l$.
    \item $c_l$: Operating cost per unit distance on link $l$.
    \item $g_l$: Cost per unit capacity on MOD access link $l$
    \item $c_i$: Cost per unit capacity at node $i$.
    \item $\delta_{il}$: Node-link incidence matrix indicator ($1$ if link $l$ enters node $i$, $-1$ if it leaves, $0$ otherwise).
    \item $V_m$: Total fleet size available for operator $m$. 
    \item $w_{o}$: Weight between mobility service and recharge service.
    \item $w_{sd}$: Mobility service dispersion weight.
    \item $w_{uc}$: Traveler utility weight.
    \item $w_{oc}$: Operator utility weight.
    \item $w_{rd}$: Recharge dispersion weight.
    \item $w_{ocr}$: Operator recharge utility weight.
    \item $w_{rs}$: Recharge station utility.
\end{itemize}

\subsection*{Decision Variables}
\begin{itemize}
    \item $p_l^{t}$: Service link pricing for all link $l \in A$ operated by $m\in M$ (\textbf{Upper level}) (note that charging pricing, i.e., cost allocation between mobility provider and charging provider, is assumed exogenous and set equal to the operating cost)
    \item $x_{s,l}^t$: Unit flow for OD $s\in \{S \cup \tilde{S}\}$ on link $l\in A$ in interval $t\in \mathcal{T}$ (\textbf{Lower level})
    \item $z_{l}^{t}$: Percent of max capacity allocated for service node $l\in A_m$ for $m\in M$ for $t\in \mathcal{T}$ (\textbf{Lower level})
    \item $\mu_{i}^{t}$: Percent of max deployed fleet size for service node $i\in N_m$ for $m \in M$ for $t \in \mathcal{T}$ (\textbf{Lower level})
    \item $u_{i}^{t}$: Percent of max capacity allocated for charging station node $i\in N_e$ for $e \in E$ for $t \in \mathcal{T}$ (\textbf{Lower level})
    \item $r_{ij}^{t,m}$: Redistribution of expected mobility service fleet (shortened as redistribution flow in the following context) between node $i$ and node $j$ for operator $m$. (\textbf{Lower level})
\end{itemize}

\begin{figure}[!htbp]
    \centering
    \includegraphics[width=0.7\textwidth]{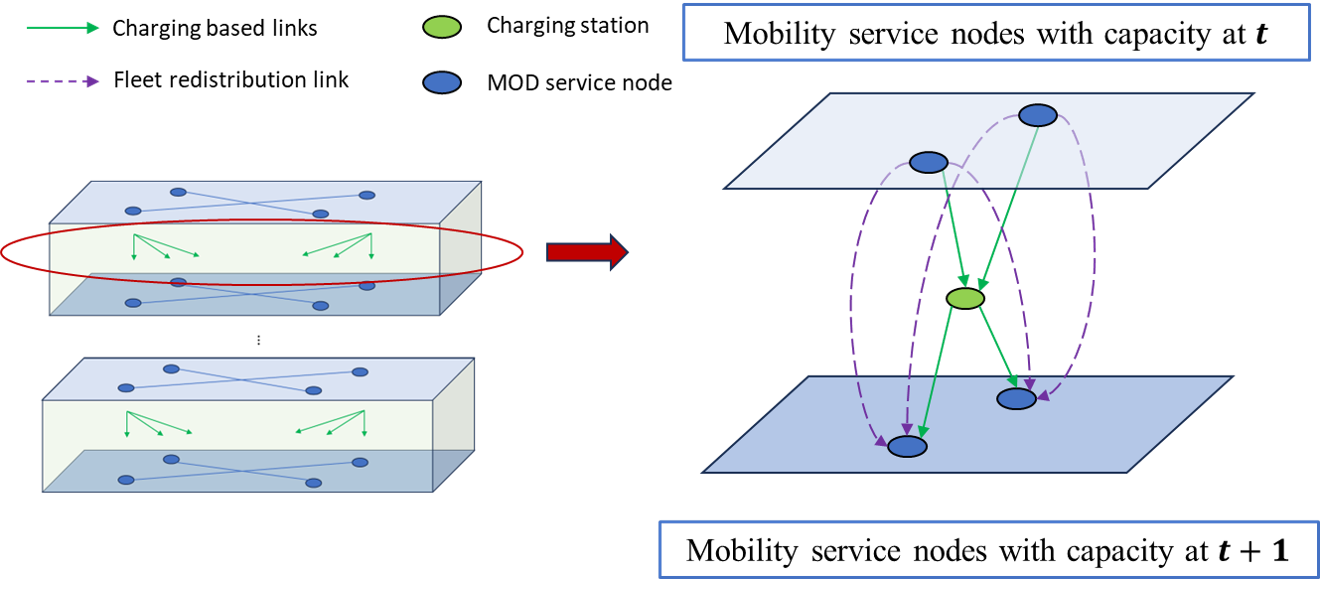}
    \caption{Illustration of recharge service subnetwork structure}
    \label{fig:charging network}
\end{figure}

\subsection{Lower level formulation}

\begin{subequations}\label{eq:lower_service}
\begin{align}
&\min_{\mathbf{x,z,u,r},\boldsymbol{\mu}} \Phi_1 = \Phi_{S} + w_o \Phi_{R} \label{eq:lower_obj} \\
&\Phi_{S} = \sum_{t \in \mathcal{T}} \Biggl[ \sum_{s \in \mathcal{S}, l \in \mathcal{A}} d_l \Bigl( w_{uc}\mathcal{U}(x_{s,l}^t) 
+ w_{sd}\mathcal{F}(x_{s,l}^t) \Bigr) \nonumber\\
&\quad + w_{oc} \sum_{s \in \mathcal{S}, l \in \mathcal{A}} d_l c_l q_s^t x_{s,l}^t + w_{oc} \sum_{m \in \mathcal{M}, l \in \mathcal{A}_m^a} g_l v_l z_l^t \Biggr] \label{eq:service_object} \\
&\Phi_{R} = \sum_{t \in \mathcal{T}} \Biggl[ \sum_{\tilde{s} \in \tilde{\mathcal{S}}, \tilde{l} \in \tilde{\mathcal{A}}} d_{\tilde{l}} \Bigl( w_{ocr}\tilde{\mathcal{U}}(x_{\tilde{s},\tilde{l}}^t) + w_{rd}\mathcal{F}(x_{\tilde{s},\tilde{l}}^t) \Bigr) \nonumber\\
&\quad + w_{rs} \sum_{e \in E, i \in \mathcal{N}_e} c_i h_i u_{i,e}^t \Biggr] \label{eq:recharge_object}\\
&\text{s.t.} \quad \sum_{l \in \mathcal{A}} \delta_{i,l} x_{s,l}^t
=
\begin{cases}
-1, & i = s^+, \\
 1, & i = s^-, \\
 0, & \text{otherwise},
\end{cases} \nonumber
\\
&\quad \forall s \in \mathcal{S},\ t \in \mathcal{T},\ i
\label{eq:service_flow}
\\
&\sum_{s \in \mathcal{S}} x_{s,l}^t q_s^t
\le
z_l^t v_l, \nonumber \\
&\quad
\forall l \in \mathcal{A}_m^a,\ m \in \mathcal{M},\ t \in \mathcal{T}
\label{eq:service_capacity}
\\
&\mu_{i=l^-}^t
\ge z_l^t,
\quad
\forall l \in \mathcal{A}_m^a,\ m \in \mathcal{M},\ t \in \mathcal{T}
\label{eq:node_deploy}
\\
&0 \le z_i^t
\le 1,
\quad
\forall i \in \mathcal{N}_m,\ m \in \mathcal{M},\ t \in \mathcal{T}
\label{eq:node_capacity}
\\
&0 \le \mu_i^t
\le 1+\epsilon,
\quad
\forall i \in \mathcal{N}_m,\ m \in \mathcal{M},\ t \in \mathcal{T}
\label{eq:node_deploy_act}
\\
&\sum_{l \in \mathcal{A}_m^a} v_l \mu_{i=l^-}^t
\le V_m,
\quad
\forall m \in \mathcal{M},\ t \in \mathcal{T}
\label{eq:fleet_cap}
\\
&V_m
\ge
\kappa
\sum_{\tau \in \mathcal{T}} \sum_{l \in \mathcal{A}_m}
\pi_t^\tau v_l z_l^\tau +
\sum_{l \in \mathcal{A}_m} v_l z_l^{t+1,m}, \nonumber \\
&\quad
\forall m \in \mathcal{M},\
t \in \mathcal{T}\setminus\{|\mathcal{T}|\}
\label{eq:buffer_shift_1}
\\
&V_m
\ge
\kappa
\sum_{\tau \in \mathcal{T}} \sum_{l \in \mathcal{A}_m}
\pi_{|\mathcal{T}|}^\tau v_l z_l^\tau +
\sum_{l \in \mathcal{A}_m} v_l z_l^{1,m}, \nonumber \\
&\quad
\forall m \in \mathcal{M}
\label{eq:buffer_shift_2}
\end{align}

\begin{align}
&0 \le x_{s,l}^t
\le 1,
\quad
\forall l \in \mathcal{A},\ s \in \mathcal{S},\ t \in \mathcal{T}
\label{eq:flow_constraint}
\\
&v_l \mu_{i=l^-}^t
=
\sum_{j \in \mathcal{N}_m} r_{i=l^-,j}^{t,m},
\nonumber\\
&\qquad
\forall l \in \mathcal{A}_m^a,\ m \in \mathcal{M},\ t \in \mathcal{T}
\label{eq:rebal_out}
\\
&\sum_{i \in \mathcal{N}_m} r_{i,j=l^-}^{t,m}
=
v_l \mu_{j=l^-}^{t+1,m},
\nonumber\\
&\qquad
\forall l \in \mathcal{A}_m^a,\ m \in \mathcal{M},\
t \in \mathcal{T}\setminus\{|\mathcal{T}|\}
\label{eq:rebal_in_1}
\\
&\sum_{i \in \mathcal{N}_m} r_{i,j=l^-}^{|\mathcal{T}|,m}
=
v_l \mu_{j=l^-}^{1,m},
\nonumber\\
&\qquad
\forall l \in \mathcal{A}_m^a,\ m \in \mathcal{M}
\label{eq:rebal_in_2}
\\
&q_{\tilde{s}}^t
=
r_{\tilde{s}^+,\tilde{s}^-}^{t,m},
\nonumber\\
&\qquad
\forall \tilde{s} \in \tilde{\mathcal{S}},\
m \in \mathcal{M},\ t \in \mathcal{T}
\label{eq:recharge_od}
\\
&r_{\tilde{s}^+,\tilde{s}^-}^{t,m}
\ge 0,
\nonumber\\
&\qquad
\forall \tilde{s} \in \tilde{\mathcal{S}},\
m \in \mathcal{M},\ t \in \mathcal{T}
\label{eq:recharge_flow_large}
\\
&\sum_{\tilde{l} \in \tilde{\mathcal{A}}}
\delta_{i,\tilde{l}} x_{\tilde{s},\tilde{l}}^t
=
\begin{cases}
-1, & i=\tilde{s}^+,\\
 1, & i=\tilde{s}^-,\\
 0, & \text{otherwise},
\end{cases}
\nonumber\\
&\qquad
\forall \tilde{s} \in \tilde{\mathcal{S}},\ t \in \mathcal{T},\ i
\label{eq:recharge_flow}
\\
&\sum_{\tilde{l} \in \tilde{\mathcal{A}}_i^a}
\sum_{\tilde{s} \in \tilde{\mathcal{S}}}
x_{\tilde{s},\tilde{l}}^t q_{\tilde{s}}^t
=
h_i u_{i,e}^t,
\nonumber\\
&\qquad
\forall i \in \mathcal{N}_e,\ e \in E,\ t \in \mathcal{T}
\label{eq:recharge_capacity}
\\
&\sum_{e \in E} \sum_{i \in \mathcal{N}_e} h_i u_{i,e}^t
=
\sum_{m \in \mathcal{M}} \sum_{\tau \in \mathcal{T}}
\sum_{l \in \mathcal{A}_m^a}
\pi_{t,\tau} v_l z_i^\tau,
\nonumber\\
&\qquad
\forall t \in \mathcal{T}
\label{eq:recharge_demand}
\end{align}

\end{subequations}

The complete lower-level model encompasses both the mobility service capacity allocation, traveler assignment, capacity redistribution, as well as charging service capacity allocation and charging assignment. The three-sided market decision-makers collectively make simultaneous matching decisions, and the coalitions of these matches are being stochastically chosen through the PURC framework as a coalitional choice model. We first present the integrated formulation for the entire lower level. The full lower-level formulation is shown in \textbf{Eq.~\ref{eq:lower_service}}.

\subsubsection{Mobility service network subproblem}

The objective function \textbf{Eq.~\ref{eq:service_object}} minimizes the total perturbed utility, representing a weighted aggregation of the disutility of the three-sided market coalitions of travelers, mobility operators, and charging operators along with the perturbed utility term across all service intervals $t$. The term $\mathcal{U}(x_{s,l}^t)$ denotes the link-additive traveler disutility calculated from flow $x_{s,l}^t$. $\mathcal{F}(x_{s,l}^t)$ is a convex perturbation function capturing the dispersion effect in route choice as established in \cite{fosgerau2022perturbed}. In this study, we use the form of $x^{2}$ as suggested in \cite{fosgerau2022perturbed} for simplicity, but it can be replaced with other recommend convex functions (e.g., $(x+1)ln(x+1)-x$). The second component accounts for the mobility providers' costs including link operating cost and access link capacity cost. All terms are scaled by their corresponding coefficient $w_{uc}$, $w_{sd}$, and $w_{oc}$. These tunable parameters balance operator and traveler disutilities, as well as the dispersion effect depicted by $\mathcal{F}(\mathbf{x})$, effectively representing a coalition as discussed in \cite{liu2024modeling} and \cite{yang2025bilevel}, but designed to capture three sides. Since the charging costs are embodied by the mobility operator in matching with the travelers, only one parameter is used for the subcoalition of mobility and charging providers. Operator disutility encompasses both link operating and fixed capacity costs: link costs are the product of unit flow $x_{s,l}^t$, demand $q_{s}^t$, link length $d_l$, and unit cost $c_l$; fixed capacity costs are derived from the decision variable $z_l^{t}$ and the access link capacity scaler $v_l$.

The constraints ensure network physical feasibility and consistent loading. \textbf{Eq.~\ref{eq:service_flow}} enforces flow conservation for every OD pair $s$ at each node $i$ during interval $t$, using the node-link incidence indicator $\delta_{i,l}$. \textbf{Eq.~\ref{eq:service_capacity}} imposes a critical access link capacity constraint, mandating that the service flow using access links $l \in \mathcal{A}_m$ does not exceed the capacity allocated by the operator. This constraint functions as the mechanism for generating endogenous congestion delays when demand approaches the allocated supply as shown in \cite{liu2024modeling} and \cite{yang2025bilevel}.

Constraints \textbf{Eq.~\ref{eq:node_deploy}--\ref{eq:buffer_shift_2}} govern the interaction between service capacity and fleet availability. \textbf{Eq.~\ref{eq:node_deploy}} defines $\mu_i^{t}$ as the physical fleet resources actually deployed to serve the access links $l\in A_m$ where node $i$ is the sink node of $l$, and it must meet or exceed the allocated service capacity $z_l^{t}$. $z_l^{t}$ captures the active fleet currently engaged in passenger service, while $\mu_i^{t}$ represents the total fleet present in the area, including cruising, idling, or charging vehicles. \textbf{Eq.~\ref{eq:node_capacity}} and \textbf{Eq.~\ref{eq:node_deploy_act}} governs the range of $\mathbf{\mu}$ and $\mathbf{z}$. $\mathbf{\epsilon}$ is an input parameter that determines the level of excessive volume allowable to stage at service nodes. \textbf{Eq.~\ref{eq:fleet_cap}} ensures that the total deployed fleet for each operator $m$ does not exceed the total fleet size $V_m$ in any interval. Finally, the system must maintain a sufficient buffer to account for vehicles unavailable due to charging. In \textbf{Eq.~\ref{eq:buffer_shift_1}-\ref{eq:buffer_shift_2}}, the first term of the Right-Hand Side (RHS) is an empirical affine function of multiple intervals of the platform's mobility demand to determine the charging demand. This represents a temporal macroscopic relationship between mobility demand and charging demand that can be empirically fitted to data. $\pi_{t}^{\tau}$ is the propagation coefficient for the interval pair $(t, \tau)$, and $\kappa$ represents the percentage of charging demand deemed unavailable for the subsequent interval. These coefficients can be calibrated to either observed or simulated data of such dynamics. The second term captures the total service capacity of mobility fleet $m$ in the subsequent interval. \textbf{Eq.~\ref{eq:buffer_shift_1}} represents when $t$ is not the last interval, and \textbf{Eq.~\ref{eq:buffer_shift_2}} reflects the cyclic manner when $t$ is the last interval. The sum of the two represents the actual fleet size that is required for the mobility service constraints, which shall be smaller than or equal to the total deployable fleet $V_m$. \textbf{Eq.~\ref{eq:flow_constraint}} restricts the normalized flow variables to the valid range $[0, 1]$.

The mobility service network subproblem adopts a structure analogous to the lower-level problem proposed in \cite{yang2025bilevel}, which demonstrates that this formulation yields a flow assignment model characterizing a traveler-operator coalition within the PURC framework. Interested readers are referred to \cite{yang2025bilevel} for the complete mathematical proof. The weights $w_{uc}$ and $w_{oc}$ governs the relative influence of each stakeholder on the resulting flow assignment behavior. Furthermore, the relative weight between $w_{sd}$ and $w_{uc}, w_{oc}$ regulates the magnitude of the dispersion effect, as established in \cite{fosgerau2022perturbed}. These three parameters can be calibrated using empirical data observed during each defined service interval.

\begin{proposition} \label{prop_1}
The mobility service network subproblem provides a standalone PURC-based flow assignment model with a traveler-operator coalition when the link utility function $u(\mathbf{x})$ is linear.
\end{proposition}

\begin{proof}Let the objective function of the service subproblem be denoted as $\Phi_1(\boldsymbol{x, z,\mu, V})$. We construct the Lagrangian $\mathcal{L}$ for the mobility service subproblem:

\begin{equation}\label{eq:ec_lagrangian_r}
\begin{aligned}
&\mathcal{L}(\boldsymbol{x, z,\mu, V, \omega, \eta, \sigma, \rho, \iota, \chi, \upsilon }) \\ 
&\quad = \Phi_1(\boldsymbol{x, z,\mu, V}) + \sum\boldsymbol{\omega}\mathcal{K}(\mathbf{x}) + \sum \boldsymbol{\eta} \mathcal{V} \left(\boldsymbol{x, z} \right) \\
&\qquad + \sum \boldsymbol{\sigma} \mathcal{R}(\boldsymbol{z, \mu}) + \sum \boldsymbol{\iota} \mathcal{B}(\boldsymbol{\mu, V}) \\
&\qquad + \sum \boldsymbol{\chi} \mathcal{D}(\boldsymbol{z, V}) + \sum \boldsymbol{\upsilon} \mathcal{P}(\boldsymbol{x, \mu, z})
\end{aligned}
\end{equation}

Where $\boldsymbol{\omega}, \boldsymbol{\eta}, \boldsymbol{\sigma}, \boldsymbol{\iota}, \boldsymbol{\chi}, \boldsymbol{\upsilon}$ are the respective dual variables for the flow conservation, service node capacity, node coefficient, total fleet size, fleet buffer, and variable bound constraints. The Karush-Kuhn-Tucker (KKT) condition with respect to $\mathbf{x}$, given by $\nabla_{\mathbf{x}} \mathcal{L} = 0$, is written as:

\begin{equation}\label{eq:ec_lagrangian_service_x}
\begin{aligned}
&\nabla_{x} \mathcal{L} = \nabla_{x}\Phi_1(x) + \omega\nabla_{x}\mathcal{K}(x) + \eta\nabla_{x}\mathcal{V}(x) + \upsilon\nabla_x\mathcal{P}(x) = 0
\end{aligned}
\end{equation}

Since $\mathbf{x}$ is linearly connected with other decision variables, the first-order derivative with respect to $\mathbf{x}$ drops $\boldsymbol{z,\mu,V}$ entirely. Because the utility function $u(x)$ is linear, Eq.~\ref{eq:ec_lagrangian_service_x} implies that the optimal flow assignment $\mathbf{x}$ is determined by the marginal disutility of traversing service network links, the difference in Lagrangian multipliers of source and sink nodes for each link, the Lagrangian multipliers of the capacity constraints $\boldsymbol{\eta}$, and the Lagrangian multiplier of link flow upper bounds. The Lagrangian multipliers $\boldsymbol{\eta}$ can be obtained using $\nabla_{\mathbf{z}} \mathcal{L} = 0$:

\begin{equation}\label{eq:ec_lagrangian_service_mu}
\begin{aligned}
&\nabla_{z} \mathcal{L}  = \nabla_{z}\Phi_1(z) + \eta\nabla_{z}\mathcal{V}(z) + \sigma\nabla_{z}\mathcal{R}(z) + \chi\nabla_{z}\mathcal{D}(z) +  \upsilon\nabla_z\mathcal{P}(z)= 0 
\end{aligned}
\end{equation}

Rearranging Eq.~\ref{eq:ec_lagrangian_service_mu} yields:

\begin{equation}\label{eq:ec_lagrangian_service_mu_shift}
\eta = -\frac{\nabla_{z}\Phi_1(z) + \sigma\nabla_{z}\mathcal{R}(z) + \chi\nabla_{z}\mathcal{D}(z) +  \upsilon\nabla_z\mathcal{P}(z)}{\nabla_{z}\mathcal{V}(z)}
\end{equation}

Eq.~\ref{eq:ec_lagrangian_service_mu_shift} represents the marginal cost of service node capacity assignment along with the Lagrangian multipliers $\boldsymbol{\sigma, \chi, \upsilon}$, which are not related to the disutility values in the objective. By combining Eq.~\ref{eq:ec_lagrangian_service_x} and Eq.~\ref{eq:ec_lagrangian_service_mu_shift}, the optimal flow is decided by the marginal disutility of link usage from the coalition of travelers and operators, along with the marginal disutility of assigning service node capacity from the operator side. This is in the same structure as shown in Yang et al. (2026). This concludes the proof.
\end{proof}

\subsubsection{Recharge and redistribution network subproblem}

The second component of the lower-level model is the Recharge and Redistribution Network Subproblem. This formulation manages the transition phase between consecutive service of each pair of intervals $t$ and $t+1$, governing both the fleet repositioning and the allocation of charging infrastructure capacity. Unlike the mobility service subproblem, where demand is exogenous (traveler requests), the demand in this subproblem is endogenous. It is derived explicitly from the fleet redistribution requirements determined by mobility operators to maintain service levels while minimizing total system disutility. The formulation for the recharge and redistribution assignment is defined by the objective function \textbf{Eq.~\ref{eq:recharge_object}} and the constraints \textbf{Eq.~\ref{eq:flow_constraint}--\ref{eq:recharge_demand}}.

\textbf{Eq.~\ref{eq:recharge_object}} presents the objective function, which minimizes the combined disutility of fleet movement and charging capacity establishment. Analogous to \textbf{Eq.~\ref{eq:service_object}}, the first term represents the link-additive disutility of routing vehicles through the recharge network, where $x_{\tilde{s},\tilde{l}}^{t}$ denotes the flow on recharge network link $\tilde{l}$ for the redistribution pair $\tilde{s}$. The second term captures the cost incurred by energy providers, where $z_{i,e}^{t}$ is the capacity decision variable for charging station $i$ managed by provider $e$. The parameters $w_{rd}, w_{ocr}, w_{rs}$ govern the weight between fleet routing, charging infrastructure investment, and route dispersion effect in the recharge network.

The constraints ensure that fleet movements are physically feasible and that energy demands are satisfied. \textbf{Eq.~\ref{eq:rebal_out} - \ref{eq:rebal_in_2}} govern the vehicle movement between mobility service node pairs in $\mathcal{N}_m$ for each operator $m$ during the recharge interval. Here, $r_{i,j}^{t,m}$ represents the redistribution flow bridging the temporal progression between service interval $t$ and $t+1$. \textbf{Eqs.~\ref{eq:rebal_out}--\ref{eq:rebal_in_2}} maintain fleet size consistency across time intervals by ensuring that vehicle inflows and outflows match the deployed fleet, $\mu_i^{t}$. Specifically, \textbf{Eq.~\ref{eq:rebal_out}} mandates that the total flow departing from service node $i$ equals the deployed fleet at $i$ during interval $t$. Similarly, \textbf{Eqs.~\ref{eq:rebal_in_1}--\ref{eq:rebal_in_2}} ensure the flow entering node $j$ satisfies the fleet deployment decision at $j$ for the subsequent interval. Notably, these conservation constraints are governed by the deployed fleet $\mu_i^{t}$ rather than the active fleet $z_l^{t}$. Because the active fleet strictly reflects MOD service demand, its sum naturally fluctuates across time intervals. The total deployed fleet, however, must remain conserved. The gap between these two values quantifies resource redundancy, representing vehicles that are idling due to supply-demand mismatches or mandatory charging.

\textbf{Eq.~\ref{eq:recharge_od}} serves as the critical coupling constraint between the mobility service and recharge subproblems. It defines the demand $q_{\tilde{s}}^{t}$ for the recharge network as exactly equal to the redistribution flow $r_{i,j}^{t,m}$ from the mobility operator's fleet management decisions.

Finally, \textbf{Eq.~\ref{eq:recharge_flow}} enforces standard flow conservation within the recharge network topology, identical in structure to \textbf{Eq.~\ref{eq:flow_constraint}}. \textbf{Eq.~\ref{eq:recharge_capacity}} constrains the flow through charging stations, ensuring that the aggregate fleet flow using charging station $i$ does not exceed the capacity $h_i u_{i,e}^{t}$ allocated by the energy provider. Finally, \textbf{Eq.~\ref{eq:recharge_demand}} enforces a system-wide energy balance, requiring that the total allocated charging capacity across all energy providers must equal the aggregate charging demand generated by the mobility services.

We use \textbf{Proposition}~\ref{prop_2} and \textbf{Proposition}~\ref{prop_3} to show the characteristics of the recharge and redistribution network subproblem formulation.

\begin{proposition} \label{prop_2}
Each unique distribution of $\mathbf{\mathbf{r}}$ with fixed $\bar{\boldsymbol{\mu}}$, $\bar{\mathbf{z}}$ input to the recharge and redistribution subproblem results in a standalone PURC-based flow assignment model with a operator-energy provider coalition when the utility function $\tilde{u(x)}$ is linear.
\end{proposition}

\begin{proof}
Let the objective function of the recharge subproblem be denoted as $\tilde{\Phi}_1(\mathbf{x, u})$. The optimization problem, conditional on fixed service decisions $\bar{\boldsymbol{\mu}}$ and $\bar{\mathbf{z}}$ along a unique distribution of $\bar{\mathbf{r}}$, reduces to minimizing $\tilde{\Phi}_1$ subject to the feasibility constraints. The corresponding Lagrangian $\tilde{\mathcal{L}}$ mimics the structure of the mobility service network subproblem's Lagrangian $\mathcal{L}$. This results in the optimal flow being decided by the coalition of marginal operator disutility of link usage and marginal energy provider disutility of allocating charging capacities, similar to Proposition 1. This concludes the proof.
\end{proof}

\begin{proposition}\label{prop_3}
Given fixed mobility service network decision variables $\bar{\boldsymbol{\mu}}$ and $\bar{\mathbf{z}}$, the binding constraints \textbf{Eq.~\ref{eq:rebal_out}--\ref{eq:recharge_od}} and \textbf{Eq.~\ref{eq:recharge_demand}} ensure that the resulting rebalancing flow vector $\mathbf{r^*}$ minimizes the operator disutility within the charging network.
\end{proposition}

\begin{proof}The recharge network demand is defined as a function of the rebalancing flow $\mathbf{r}$. Specifically, the charging OD demand $\mathbf{\tilde{q}} = \mathcal{M}(\mathbf{r})$, where $\mathcal{M}$ is the incidence mapping relating redistribution OD pairs to flow variables. Consequently, the flow conservation constraints and capacity constraints become implicitly dependent on $\mathbf{r}$. In addition, the system-wide capacity constraint becomes only $\mathbf{z}$ dependent. The optimality of the system therefore only depends on variables ($\mathbf{x, u, r}$) with $\mathbf{r}$ being embedded inside the Lagrangian $\tilde{L}$. As a result, the binding constraints compel the rebalancing flow $\mathbf{r}$ along with $\mathbf{x}$ and $\mathbf{u}$ to distribute fleet size in a manner that minimizes the operator's generalized disutility in the charging network. This concludes the proof.
\end{proof}

The recharge network subproblem is nonconvex when ($\mathbf{x, u, r}$) are solved simultaneously. This is caused by the product of $\mathbf{x}$ and $\mathbf{q}$ in \textbf{Eq.~\ref{eq:recharge_capacity}}. This results in the combined lower level problem becoming a nonconvex Quadratic Program with Quadratic Constraints (QPQC). 

\subsection{Upper-Level Formulation}

The upper-level optimization problem represents the profit maximization strategy of the MOD service providers. In this leader-follower framework, operators determine the optimal link-based service pricing $\mathbf{p}$, anticipating the equilibrium response of travelers and operator decisions in the lower level. The formulation is shown in \textbf{Eq.~\ref{eq:upper_level}}.

The objective function \textbf{Eq.~\ref{eq:upper_obj}} maximizes the total system profit across all time intervals $t \in \mathcal{T}$ and operators $m \in \mathcal{M}$. The first and second term capture the service operation profit, which is calculated by the total service income minus the link operating cost and capacity allocation cost. The third term represents the charging related cost. In this way, the upper level objective function optimizes the MOD service profit that also considers the charging service. Welfare-based objectives can also be used for public platforms.

Constraints \textbf{Eq.~\ref{eq:upper_price}} regulate link based service prices $p_l^{t}$ must be less than the price cap $\hat{p_l}$. \textbf{Eq. \ref{eq:upper_op_prof}} represents the IR constraint for the MOD operator’s budget balance, ensuring that each operator $m$ remains profitable (or above a minimum cost threshold for public operators) while participating in the platform. The variables $\mathbf{x,z, u}$ are determined from the lower-level PURC assignment ($\Phi_1$), representing the optimal response on both traveler, operator, and energy provider to the upper-level pricing and capacity decisions.

\begin{subequations}\label{eq:upper_level}
\begin{align}
&\max_{\mathbf{p}} \Phi_0 = \sum_{t \in \mathcal{T}} \sum_{m \in \mathcal{M}} \sum_{s \in \mathcal{S}} \sum_{l \in \mathcal{A}_m} d_l (p_l^{t} - c_{l}) q_s^t x_{s,l}^t \notag \\
& \quad - \sum_{t\in \mathcal{T}} \sum_{m \in \mathcal{M}} \sum_{l \in \mathcal{A}_m^a} g_{l} v_l z_l^{t} \notag \\
& \quad - \sum_{t\in \mathcal{T}} \sum_{e \in E} \sum_{i \in \mathcal{N}_e}c_i h_i u_i^{t} \label{eq:upper_obj}\\
&\text{Subject to:} \notag \\
& 0 \leq p_l^{t} \leq \hat{p}_l \quad \forall l \in \mathcal{A}_m, \forall t \in \mathcal{T}, \forall m \in \mathcal{M} \label{eq:upper_price}\\
& 0 \leq \sum_{t \in \mathcal{T}} \sum_{s \in \mathcal{S}} \sum_{l \in \mathcal{A}_m}d_l(p_l^t-c_l)q_s^t x_{s,l}^t - \sum_{t\in \mathcal{T}} \sum_{l \in \mathcal{A}_m^a} g_{l} v_l z_l^{t} \nonumber \\
&\quad \forall m \in \mathcal{M} \label{eq:upper_op_prof}\\
&(\mathbf{x}, \mathbf{z}, \mathbf{u}) = \arg \min \Phi_1(\mathbf{p}) \label{eq:upper_lower}
\end{align}
\end{subequations}

\subsection{Proposed heuristic}
\subsubsection{Lower Level Lexicographic Decomposition}

\begin{figure}[htbp]
    \centering
    \includegraphics[width=0.6\textwidth]{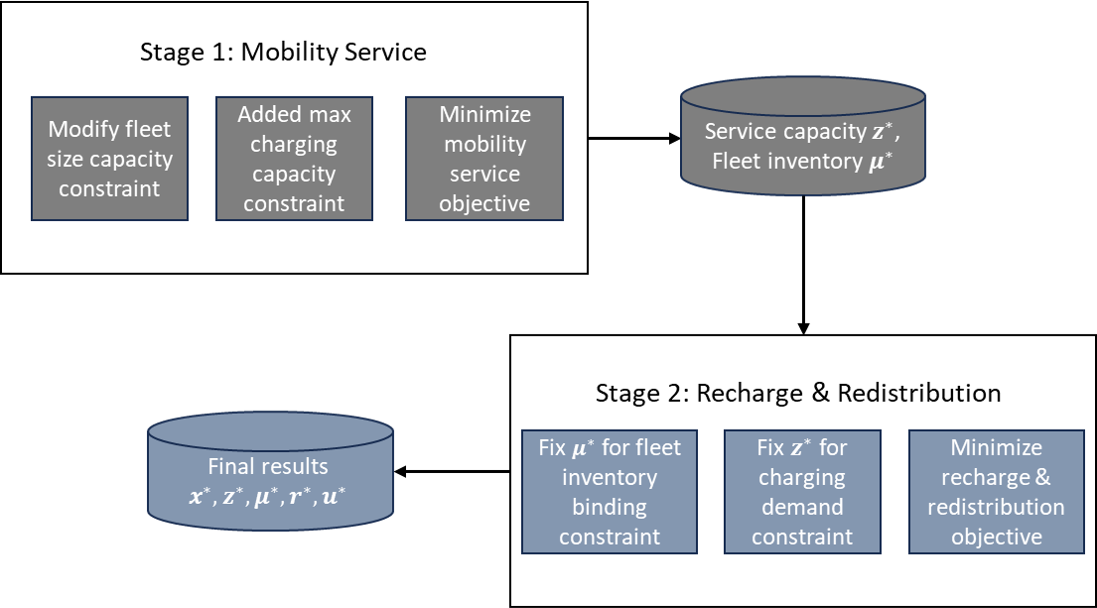}
    \caption{Decomposition heuristic flowchart}
    \label{fig:decomposition}
\end{figure}

The nonconvexity in the lower level QPQC problem leads to excessive computational complexity when facing larger-scale instances. To overcome this bottleneck, we propose a lexicographic decomposition heuristic strategy to tackle the complexity of the integrated lower-level model. This approach decouples the \textit{Service Assignment Problem} (Stage 1) from the \textit{Recharge and Redistribution Problem} (Stage 2), solving them sequentially as illustrated in \textbf{Fig.~\ref{fig:decomposition}}. The optimal fleet deployment states determined in Stage 1 serve as binding constraints for Stage 2. 

\paragraph{Stage 1: Mobility Service Optimization}\mbox{}

In the first stage, we solve the service network subproblem $\Phi_{S}$ in isolation, incorporating additional constraints to prevent feasibility violations in the subsequent stage. The decision variables are the passenger flows $\mathbf{x}$, service capacities $\mathbf{z}$, and fleet deployment $\boldsymbol{\mu}$.

We modify constraint \textbf{Eq.~\ref{eq:fleet_cap}} to \textbf{Eq.~\ref{eq:stage1_fleet}} to ensure that the total fleet size of each operator $m$ remains consistent across $\mathcal{T}$; otherwise, the redistribution flow balance in \textbf{Eq.~\ref{eq:rebal_out} - \ref{eq:rebal_in_2}} would be violated when passing $\boldsymbol{\mu}$ to the subsequent step. 

\begin{equation}
\sum_{l \in \mathcal{A}_m^a} v_l \cdot \mu_{i=l^-}^{t}  = V_m \quad \forall m \in \mathcal{M}, t \in \mathcal{T} \label{eq:stage1_fleet}
\end{equation}

Constraint \textbf{Eq.~\ref{eq:recharge_cap}} is introduced to ensure that the total charging demand derived from the regression function does not exceed the maximum allowable charging capacity.

\begin{equation}
\sum_{m \in \mathcal{M}} \sum_{\tau \in \mathcal{T}} \sum_{l \in \mathcal{A}_m^a} \pi_{t}^{\tau} v_l z_l^{\tau} \le \sum_{e \in E} \sum_{i \in \mathcal{N}_e}v_{i} \quad \forall t \in \mathcal{T} \quad \label{eq:recharge_cap}
\end{equation}

\paragraph{Stage 2: Recharge and Redistribution Optimization}\mbox{}

We fix $\mathbf{z}^*$ and $\boldsymbol{\mu}^*$ obtained from stage 1 and treat them as input parameters. In the second stage, we optimize the recharge and redistribution network subproblem $\Phi_{R}$. This stage determines the redistribution flows $\mathbf{r}$, the network flow assignment $\tilde{\mathbf{x}}$, and recharge capacities $\mathbf{u}$.

\textbf{Eq.~\ref{eq:stage2_supply} - \ref{eq:stage2_demand_2}} explicitly couple the two stages, ensuring that the redistribution flows $\mathbf{r}$ match the fleet inventory changes dictated by the optimal service solution $\boldsymbol{\mu}^*$. \textbf{Eq.~\ref{eq:stage2_recharge_demand}} uses $\mathbf{z^*}$ to regulate the total charging capacity.

\begin{subequations} \label{eq:decomp_stage2}
\begin{align}
& (\mathbf{r}^*, \tilde{\mathbf{x}}^*, \mathbf{u}^*) = \arg \min_{\mathbf{r}, \tilde{\mathbf{x}}, \mathbf{z}_{e}} \Phi_{R} \\
& \qquad \text{Subject to:} \notag \\
& \text{Constraints } \eqref{eq:recharge_flow} - \eqref{eq:recharge_capacity} \notag \\
& \sum_{j \in \mathcal{N}_m} r_{i=l^-,j}^{t,m} = v_l \cdot {\mu^{t*}_{i=l^-}} \quad \forall l \in \mathcal{A}_m^a, t \in \mathcal{T} \label{eq:stage2_supply} \\
& \sum_{i \in \mathcal{N}_m} r_{i,j=l^-}^{t,m} = v_{j=l^-} \cdot {\mu^{^{t+1,m}}_{j=l^-}}^{*} \quad \forall l \in \mathcal{N}_m, t \in \mathcal{T}/|\mathcal{T}|\label{eq:stage2_demand_1} \\
& \sum_{i \in \mathcal{N}_m} r_{i,j}^{|\mathcal{T}|,m} = v_{j=l^-} \cdot {\mu^{1,m}_{j=l^-}}^{*} \quad \forall j \in \mathcal{N}_m \label{eq:stage2_demand_2}\\
& \sum_{m \in \mathcal{M}} \sum_{\tau \in \mathcal{T}} \sum_{l \in \mathcal{A}_m^a} \pi_{t, \tau} v_l {z_l^{\tau}}^* \le \sum_{e \in E} \sum_{i \in \mathcal{N}_e} h_{i}u_{i,e}^{t} \quad \forall t \in \mathcal{T} \label{eq:stage2_recharge_demand}
\end{align}
\end{subequations}

Proposition~\ref{prop_4} shows that, under the appropriate conditions, the optimality gap of the results relative to solving \textbf{Eq.~\ref{eq:lower_service}} directly can be bounded.

\begin{proposition}\label{prop_4}
The sequential decomposition method yields solutions that are $\epsilon$-optimal with respect to the integrated model in \textbf{Eq.~\ref{eq:lower_service}} if the marginal utility of the service network assignment strictly exceeds that of the associated fleet recharge and redistribution assignment.
\end{proposition}

\begin{proof}The Service Dominance Condition requires that the marginal utility (or gradient magnitude) of the service network assignment strictly exceeds that of the associated fleet recharge and redistribution assignment. Since $\boldsymbol{\mu}^*$ and $\boldsymbol{z}^*$ are the only coupling variables shared across the two stages, the condition is formally stated as:
\begin{equation}
\left\| \frac{\partial \Phi_S}{\partial \boldsymbol{\mu}} \right\| \gg \left\| \frac{\partial \Phi_R}{\partial \boldsymbol{\mu}} \right\|,\ 
\left\| \frac{\partial \Phi_S}{\partial \boldsymbol{z}} \right\| \gg \left\| \frac{\partial \Phi_R}{\partial \boldsymbol{z}} \right\|
\end{equation}
Under this condition, the optimal deployment strategy is driven primarily by the optimization of the service network subproblem, rendering the redistribution and recharge disutility surface flat relative to the service disutility surface. Consequently, the optimal trajectory of the combined objective function $\Phi_{1} = \Phi_S + \Phi_R$ aligns asymptotically with the trajectory of $\Phi_S$, justifying the lexicographic solving approach. \end{proof}

The condition shown in Proposition~\ref{prop_4} typically holds true for the majority of real-world applications. We can also apply such condition when solving the integrated lower level model by assigning sufficiently small values to the weight ($w_{o}$) in \textbf{Eq.~\ref{eq:recharge_object}}.

The stage 2 model is still a QPQC. To further accelerate the solution speed, we introduce the Alternating Minimization (AM) heuristic. The AM heuristics systematically avoids the computational bottleneck by decoupling the optimization into two alternating steps. First, by holding $\tilde{\mathbf{x}}$ constant, the charging node capacity constraint \textbf{Eq.~\ref{eq:recharge_capacity}} is strictly linearized, reducing the optimization of the redistribution flows $\mathbf{r}$ to an efficient Linear Program (LP). In this step, the flow conservation constraint \textbf{Eq.~\ref{eq:recharge_flow}} is not considered. Subsequently, by fixing the newly updated $\mathbf{r}$, the model reduces to a standard Quadratic Program (QP) to solve for the flow assignment problem. By iteratively alternating between these two strictly convex subproblems until convergence, the proposed AM heuristics ensures tractability and rapid execution for large-scale problems. The full heuristic is presented in \textbf{Algorithm \ref{alg:stage2_am}}.

\begin{algorithm}
\caption{Alternating Minimization Heuristic for Stage 2}
\label{alg:stage2_am}
\begin{algorithmic}[1]
\Require Stage 1 optimal parameters $\boldsymbol{\mu}^*$, $\mathbf{z}^*$, Tolerance $\epsilon_{bcd}$, Max Iterations $K_{AM}$.
\State \textbf{Initialization:} 
\State \quad Initialize routing fractions $\tilde{\mathbf{x}}^{(0)}$ uniformly (e.g., $\tilde{x}_{\tilde{s}\tilde{l}}^{(0)} = 1 / |\tilde{\mathcal{A}}_{\tilde{s}}|$).
\State \quad Set iteration counter $k \gets 0$, and $\Phi_R^{(0)} \gets \infty$.

\While{$k < K_{AM}$}
    \State \textbf{Step A: Redistribution Flow Optimization ($\mathbf{r}$-Step)}
    \State \quad Fix $\tilde{\mathbf{x}} \gets \tilde{\mathbf{x}}^{(k)}$.
    \State \quad Solve for $\mathbf{r}^{(k+1)}$ and $\mathbf{u}^{(k+1)}$ by minimizing $\Phi_R$ subject to:
    \State \quad \quad Fleet inventory constraints \eqref{eq:stage2_supply} - \eqref{eq:stage2_demand_2}.
    \State \quad \quad Linearized charging constraint: $\sum (\tilde{x}^{(k)} \cdot r) = u \cdot v_i$.
    
    \State \textbf{Step B: Recharge Routing Optimization ($\tilde{\mathbf{x}}$-Step)}
    \State \quad Fix $\mathbf{r} \gets \mathbf{r}^{(k+1)}$.
    \State \quad Solve for $\tilde{\mathbf{x}}^{(k+1)}$ and $\mathbf{u}^{(k+1)}$ by minimizing $\Phi_R$ subject to:
    \State \quad \quad Flow conservation \eqref{eq:recharge_flow}.
    \State \quad \quad Linearized charging constraint: $\sum (\tilde{x} \cdot r^{(k+1)}) = u \cdot h$.

    \State \textbf{Convergence Check:}
    \State \quad Evaluate current objective $\Phi_R^{(k+1)}$.
    \If{$| \Phi_R^{(k)} - \Phi_R^{(k+1)} | \le \epsilon_{bcd}$}
        \State \textbf{break}
    \EndIf
    \State \quad $k \gets k + 1$
\EndWhile
\State \Return Local optimal variables $(\mathbf{r}^*, \tilde{\mathbf{x}}^*, \mathbf{u}^*) \gets (\mathbf{r}^{(k)}, \tilde{\mathbf{x}}^{(k)}, \mathbf{u}^{(k)})$.
\end{algorithmic}
\end{algorithm}

\subsection{Bilevel model solution method}

The bilevel problem \textbf{Eq.~\ref{eq:upper_level}} is a nonconvex function because of the complex interaction between the upper and lower level problems, along with the nonconvexity of the lower-level problem itself. To address this, we employ a Projected Gradient Ascent (PGA) algorithm coupled with a Finite Difference approximation scheme. This approach treats the lower-level problems as a "black box" that returns the system state ($\mathbf{x}, \mathbf{r}, \mathbf{z}$) for any given price vector.

The solution procedure iterates through three key phases: gradient estimation, step size determination, and projection. Since analytical gradients are unavailable, we estimate the sensitivity of the objective function $\nabla \Phi_0$ numerically. To handle the strict budget constraint (\textbf{Eq.~\ref{eq:upper_op_prof}}), we incorporate it into the upper-level objective using a penalty method. Let $\gamma$ denote the penalty coefficient and $\mathcal{B}$ represent the magnitude of the budget violation (where $\mathcal{B} = 0$ if the constraint is satisfied). We define the augmented objective function as $\Phi_0^{\prime} = \Phi_0 - \gamma \mathcal{B}$.

For each pricing decision $p_{t}$, we apply a small perturbation $\delta$ and solve the complete lower-level system to observe the resulting change in profit. The partial derivative is approximated as $\frac{\partial \Phi_0^{\prime}}{\partial p_{t}} \approx \frac{\Phi_0^{'}(\mathbf{p} + \delta) - \Phi_0^{'}(\mathbf{p})}{\delta}$. The decomposition algorithm can be applied to solve the lower-level problem efficiently, allowing quick capture of the aggregate response of the user equilibrium and fleet redistribution logic to price signals.

To ensure practical feasibility, the candidate price vector $\tilde{\mathbf{p}}$ is projected onto the feasible set $\Omega$ at each iteration. This projection operator $\Pi_{\Omega}(\cdot)$ enforces box constraints, ensuring prices remain within a pre-defined valid range $[\underline{p}, \bar{p}]$, and ramp constraints, which cap the change in price between iterations to prevent market shocks and ensure algorithmic stability. Furthermore, to guarantee monotonic convergence to a local optimum, we employ a backtracking line search strategy. Instead of a fixed step size, we dynamically adjust the learning rate $w_{uc}$. A candidate step is accepted only if it satisfies the Armijo condition, which ensures that the improvement in profit is sufficiently proportional to the gradient magnitude: $\Phi_0^{'}(\mathbf{p}^{(k)} + w_{uc} \mathbf{d}_k) \ge \Phi_0^{'}(\mathbf{p}^{(k)}) + c_1 w_{uc} \nabla \Phi_0^{'}(\mathbf{p}^{(k)})^T \mathbf{d}_k$. If the condition is not met, the step size $w_{uc}$ is reduced by a decay factor $\beta$ until the condition holds or the step becomes negligible.

To expand the solution strategy to a global optimum, we employ a Multi-Start Initialization Strategy. The optimization algorithm is executed in parallel, each initialized with a distinct pricing vector $\mathbf{p}^{(0)}$. This parallelized exploration of the solution space reduces the likelihood of trapping in local optima. Upon convergence of all instances, we select runs that do not violate budget constraints ($\mathcal{B}=0$). The solution yielding the highest system profit $\Phi_0^*$ is selected as the global output. If none of the converged results meet the budget constraints, the platform is deemed unsustainable. The complete solution procedure is summarized in \textbf{Algorithm \ref{alg:pricing_optimization_decomposition}}.

\renewcommand{\thealgorithm}{2}
\begin{algorithm}
\caption{Multi-Start Projected Gradient Ascent with Lexicographic Decomposition}
\label{alg:pricing_optimization_decomposition}
\begin{algorithmic}[1]
\Require bounds $[\underline{p},\bar{p}]$, tolerance $\epsilon_{rel}$, max iterations $K$, perturbation $\delta$, penalty $\gamma$, number of starts $N_{\mathrm{starts}}$

\Function{PenalizedProfit}{$\mathbf{p}$}
    \State Solve Service Subproblem $\to (\mathbf{x}^*,\mathbf{z}^*,\boldsymbol{\mu}^*)$
    \State Solve Recharge Subproblem given $(\mathbf{z}^*,\boldsymbol{\mu}^*)$
    \Statex \hspace{\algorithmicindent} $\to (\mathbf{r}^*,\tilde{\mathbf{x}}^*,\mathbf{u}^*)$ using Algorithm~\ref{alg:stage2_am}
    \State Evaluate profit $\Phi_0(\mathbf{p})$ and budget violation $\mathcal{B}(\mathbf{p})$
    \State \Return $\Phi_0'(\mathbf{p})=\Phi_0(\mathbf{p})-\gamma\max\{\mathcal{B}(\mathbf{p}),0\}$
\EndFunction

\For{$run=1$ to $N_{\mathrm{starts}}$} \textbf{in parallel}
    \State Initialize $\mathbf{p}^{(0)} \sim \mathrm{Uniform}(\underline{p},\bar{p})$, $k \gets 0$
    \State $\Phi_0'(\mathbf{p}^{(0)}) \gets \Call{PenalizedProfit}{\mathbf{p}^{(0)}}$
    \While{$k<K$ and not converged}
        \ForAll{$t \in \mathcal{T}$, $m \in \mathcal{M}$}
            \State $\mathbf{p}' \gets \mathbf{p}^{(k)}$; $p_t' \gets p_t^{(k)}+\delta$
            \State $\nabla_t \Phi_0' \gets \dfrac{\Phi_0'(\mathbf{p}')-\Phi_0'(\mathbf{p}^{(k)})}{\delta}$,
            where $\Phi_0'(\mathbf{p}')=\Call{PenalizedProfit}{\mathbf{p}'}$
        \EndFor
        \State $w_{uc} \gets 1/\max |\nabla \Phi_0'|$
        \Repeat
            \State $\tilde{\mathbf{p}} \gets \Pi_{\Omega}\!\left[\mathbf{p}^{(k)}+w_{uc}\nabla \Phi_0'\right]$
            \State Evaluate $\Phi_0'(\tilde{\mathbf{p}})=\Call{PenalizedProfit}{\tilde{\mathbf{p}}}$
            \If{Armijo condition fails}
                \State $w_{uc} \gets \beta w_{uc}$
            \EndIf
        \Until{Armijo condition holds or $w_{uc}$ is too small}
        \State $\mathbf{p}^{(k+1)} \gets \tilde{\mathbf{p}}$, $k \gets k+1$
        \State Check relative convergence using $\epsilon_{rel}$
    \EndWhile
    \State Record $\mathbf{p}_{run}^* \gets \mathbf{p}^{(k)}$, $\Phi_{run}^* \gets \Phi_0(\mathbf{p}^{(k)})$, and $\mathcal{B}(\mathbf{p}^{(k)})$
\EndFor
\State Select $n=\arg\max_{run}\{\Phi_{run}^*:\mathcal{B}(\mathbf{p}_{run}^*)=0\}$
\State \Return $\mathbf{p}_n^*$ and $\Phi_n^*$
\end{algorithmic}
\end{algorithm}

\section{Numerical experiments}

We begin by validating the proposed heuristic for solving the lower-level model. Using a synthetic network and simulated demand patterns, we benchmark the exact solutions against those generated by the decomposition and AM heuristic to demonstrate the method's computational efficiency and accuracy. Afterwards, we conduct a series of numerical experiments to validate the proposed bilevel eMaaS assignment framework and evaluate its analytical properties. The current experimental setup assumes a monopolistic market structure featuring a single MOD service provider and one charging provider operating within the platform. Extending the analysis to capture the complex competitive dynamics among multiple MSPs is reserved for future research. 

All computational instances are executed on a workstation equipped with an Intel Core i7-13705H processor and 32 GB of RAM. For brevity in the subsequent discussions, the "recharge and redistribution network subproblem" is referred to simply as the "recharge network subproblem."

\subsection{Evaluation of the proposed heuristic}

To validate the proposed heuristic for the lower level problem, we construct a synthetic service network with 65 nodes and 540 links and solve only that lower level minimization problem. The test instance includes two time intervals with simulated demand patterns. In addition, the recharge network contains 570 inter-layer links connecting the service layers across the two time intervals. The data for the test case is available in project's GitHub repository \citep{yang2026three_sided}.

All solution approaches are evaluated under the same parameter settings. In particular, the weight parameter of recharge objective is set to $w_o = 0.01$, which satisfies the condition stated in \textbf{Proposition 4}. We first use Gurobi to solve the integrated lower-level model directly within a 10-minute time limit, then tackle the same problem using the proposed decomposition and AM heuristics.

As shown in \textbf{Fig.~\ref{fig:heuristic_obj}}, the AM heuristic quickly converges and reaches the stopping tolerance of $10^{-4}$ within 8 iterations. Most of the improvement occurs during the first three iterations, after which the stage-2 objective value stabilizes, indicating good numerical convergence behavior. The final results are summarized in \textbf{Table~\ref{tab:solution_comparison}}. The decomposition and AM heuristics achieve an objective value of 455.80, which is lower than the benchmark solution value of 455.91 achieved within the 10-minute timeframe. Meanwhile, the proposed approach reduces the computation time from 600 seconds to 13.2 seconds, a reduction rate of 2.2\% of the runtime. These results demonstrate that the proposed heuristic can produce near-exact solutions with substantially improved computational efficiency.

\begin{figure}[htbp] 
\centering 
\includegraphics[width=0.6\textwidth]{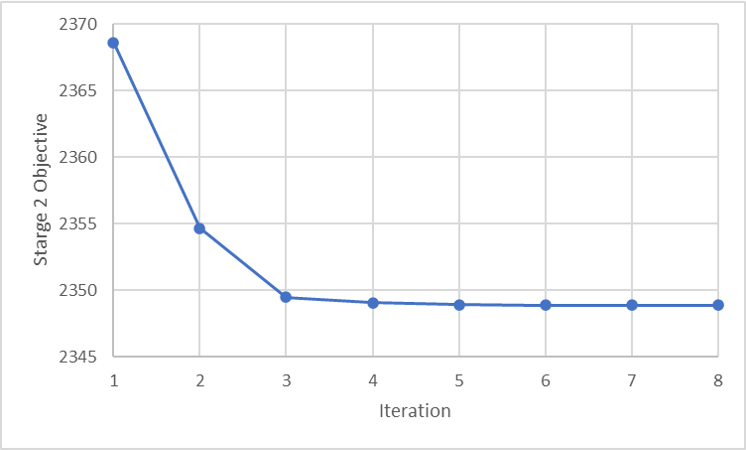} \caption{Convergence of the stage-2 objective value under the AM heuristics} \label{fig:heuristic_obj} 
\end{figure} 

\begin{table}[htbp] 
\centering 
\caption{Comparison between the exact solution and the proposed decomposition and AM heuristics.} \label{tab:solution_comparison} \small 
\begin{tabularx}{\linewidth}{l>{\centering\arraybackslash}Xcc} \hline & Objective ($\Phi_S + 0.01\Phi_R$) & Optimality gap & Run time (second) \\ 
\hline Exact solution via Gurobi & 455.91 & 0.13\% & 600 \\ 
Decomposition \& AM heuristic & 455.8 & - & 13.2 \\ \hline \end{tabularx} 
\end{table}

\subsection{Case study to illustrative analytical insights from model}

We implement the model using the Nguyen-Dupuis network. The analysis begins by isolating the lower-level model to assess system performance under fixed service prices and a predefined total fleet size. We first establish a baseline scenario to explore flow assignment and capacity allocation across the service and recharge networks. We then conduct two sensitivity analyzes: the first evaluates the impact of power grid constraints by varying charging infrastructure limits, and the second examines how the route dispersion function influences network equilibrium and capacity allocations.

We then evaluate the complete bilevel formulation to explore the strategic interplay between long-term total fleet sizing, dynamic service pricing, and operational routing decisions. Because of the small size of the instance, we can solve the integrated lower-level problem directly with Gurobi in this section to guarantee an exact solution for the lower level. The results illustrate insights that can be gained about the trade-offs made in operating such a mobility-energy platform.

\subsubsection{Nguyen-Dupuis Network}

The original Nguyen-Dupuis network comprises 13 nodes and 19 directional links. We expand it to make the 19 links bidirectional and use it as the OOP subnetwork. To represent the added MOD service, we add another MOD subnetwork based on nodes 2, 5, 6, 7, 8, 9, 10, and 11, as shown in the shaded area of \textbf{Fig.~\ref{fig:np_network_topology}}. We label the MOD service nodes as 102, 105, 106, 107, 108, 109, 110, and 111 to distinguish them from the OOP network. All MOD nodes and links are operated by a single operator. Therefore, only one MOD subnetwork is created. The MOD service subnetwork has a higher link service fee compared to the OOP network, while providing significantly shorter link travel times. Specifically, the MOD links operate at a speed of 25 miles/h compared to 15 miles/h on the base links, reflecting the premium nature of the service. As a result, the mobility service network contains 21 nodes and 74 links.

\begin{figure}[htbp]
    \centering
    \includegraphics[width=0.6\textwidth]{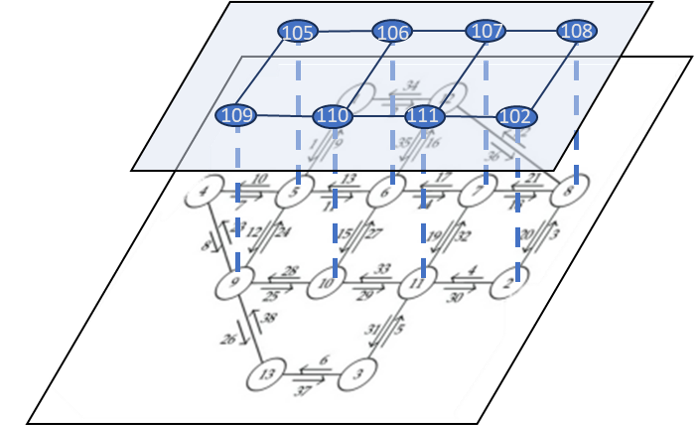}
    \caption{Expanded Nguyen-Dupuis service network}
    \label{fig:np_network_topology}
\end{figure}

Two charging stations are placed at Node 7 and Node 9, both operated by one provider. The recharge network is constructed by directly linking each node pair representing a redistribution movement between two subsequent time intervals. By adding direct links connecting each node from interval $t$ to the charging node, and links between the charging node to each node in subsequent interval $t+1$, we create a recharge network containing 96 links and 18 nodes. \textbf{Fig.~\ref{fig:recharge_topology}} illustrates the links out from node 109 at interval $t$ to all service nodes in interval $t+1$ as an example. We label the charge nodes as 207 and 209 respectively.

\begin{figure}[htbp]
    \centering
    \includegraphics[width=0.6\textwidth]{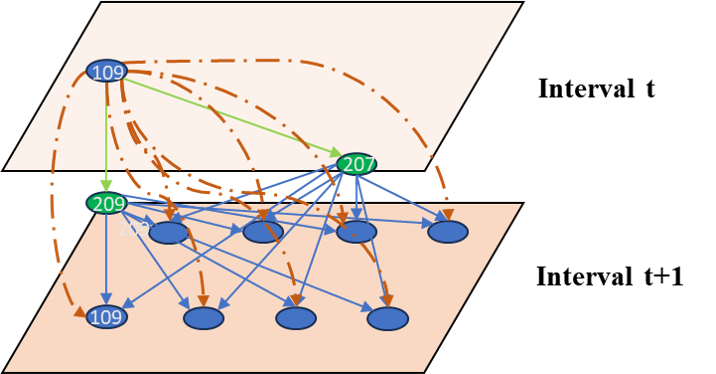}
    \caption{Recharge network illustration using node 109 as an example}
    \label{fig:recharge_topology}
\end{figure}

\subsubsection{Baseline scenario}

We establish the baseline scenario that serves as the foundation for the subsequent numerical experiments. All input parameters defined here are held constant throughout the sensitivity analyses, unless explicitly treated as the variable of interest in later sections. All detailed input settings, including the complete passenger demand matrices and system parameters, are provided in \textbf{Appendix A}.

We involve three intervals ($t=1, 2, 3$) in a cyclic manner for our numerical experiments. Nodes 1, 2, 3, and 4 serve as both origin and destination centroids with OD demand exhibiting distinct temporal distributions. As detailed in \textbf{Table \ref{tab:demand_data}}, demand peaks shift significantly across intervals; for instance, traffic from Origin 1 to Destination 3 drops from 1200 in $t=1$ to 360 in $t=2$, then stabilizes at 240 in $t=3$, requiring the system to dynamically adjust fleet distribution.

The input parameters for the integrated mobility-energy platform model are summarized in \textbf{Table \ref{tab:np_input}}. The traveler disutility function $\mathcal{U}_{l}$ in the mobility service network subproblem is modeled as a generalized cost combining travel time and link price:
\begin{equation}
    \mathcal{U}_{l} = \beta_{VOT} \cdot t_{l} + p_{l}
\end{equation}
where $t_{l}$ denotes the link travel time and $\beta_{VOT}$ represents the traveler's Value of Time, assumed to be \$20/h. The term $p_{l}$ refers to the link-specific pricing. For the MOD subnetwork, pricing is distance based at a rate of \$0.5 per unit length. In contrast, the OOP network uses a fixed pricing structure, charging a "flat fee" of \$1 per link traversal regardless of its length. We use this "flat fee" to reflect the combined fixed and variable costs associated with modes other than the MOD service. 

The operator's disutility $\tilde{\mathcal{U}}_{\tilde{l}}$ in the recharge network is defined by operational distances and facility access fees:
\begin{equation}
    \tilde{\mathcal{U}}_{\tilde{l}} = c_{\tilde{l}} d_{\tilde{l}} + \delta_{\tilde{l} \in \tilde{A^a_i}} p_{station}
\end{equation}
where $c_{\tilde{l}}$ is the operational cost per unit length, $d_{\tilde{l}}$ is the length of the link length, $p_{station}$ is the fixed charging price, and $\delta_{\tilde{l} \in \tilde{A^a_i}}$ is a binary indicator that equals 1 if the link involves accessing a charging station.

Regarding the coalition weights, we set the traveler utility weight $w_{uc}$ to 1 and the operator utility weight $w_{oc}$ to $5 \times 10^{-4}$. Since traveler disutility is captured by unit value while operator disutility is captured by aggregated cost, the operator weight must be sufficiently small to balance the two objectives. The value of $5 \times 10^{-4}$ effectively normalizes the operator's influence to approximately 0.225 of the average weight relative to the traveler, prioritizing traveler choice in the equilibrium.

In this initial baseline scenario, only one operator provides the MOD service, and the total fleet size is 1600. Both service nodes and charging stations are subject to a capacity limit of 300 units per interval. The operational cost is set to \$0.2 per unit length for MOD links and \$0.1 for recharge network links, while other links having zero cost. Charging prices are location-dependent to test spatial sensitivity, with the station at node 7 charging \$4 per usage and the station at node 9 charging \$2 per usage. Finally, the temporal propagation of charging load is governed by the empirical multiperiod function with factors $[0.05, 0.1, 0.2]$, meaning that 5\%, 10\%, and 20\% of the service demand from previous intervals ($t, t-1, t-2$) propagates into the current charging load, simulating a lag effect between mobility demand and charging demand. Let $\tilde{\mathcal{D}}_t$ and $\mathcal{D}_t$ denote the charging and service demands, respectively. \textbf{Eq.~\ref{eq:charging_propagation}} illustrates the calculation. We solve the model using these input parameters and use the results as the benchmark for further sensitivity analysis.

\begin{equation}\label{eq:charging_propagation}
\begin{bmatrix}
\tilde{\mathcal{D}}_1 \\
\tilde{\mathcal{D}}_2\\
\tilde{\mathcal{D}}_3
\end{bmatrix}
=
\begin{bmatrix}
0.05 & 0.20 & 0.10 \\
0.10 & 0.05 & 0.20 \\
0.20 & 0.10 & 0.05
\end{bmatrix}
\begin{bmatrix}
\mathcal{D}_1\\
\mathcal{D}_2\\
\mathcal{D}_3
\end{bmatrix}.
\end{equation}


We solve the base case model using both the exact approach and the proposed heuristics for further validation. We achieve global optimality using Gurobi in just 4 seconds, demonstrating robust computational efficiency for moderate scale instances. By comparison, implementing the proposed decomposition and AM heuristic yields a solution with an objective value within a 0.01\% margin of the global optimum in 0.3 seconds. This further demonstrates the effectiveness of the proposed heuristic. For all subsequent analyses, the moderate scale of the test instances is small enough to solve the integrated lower-level model directly using Gurobi.

The eMaaS system reaches an optimal operational state with a total objective value of 1,463.20. The maximum active fleet size is 1,473 vehicles. In other words, during the system's peak demand window, 1,473 out of the 1,600 total available vehicles are actively deployed to satisfy passenger mobility, leaving the remainder strategically buffered for concurrent charging or simply idling due to demand-supply mismatch.

The mobility service network exhibits significant variability in access link capacity. \textbf{Fig.~\ref{fig:service_util_heatmap}} illustrates the utilization rates of access links $z_l$, which are labeled by the linked service nodes. These values reflect the active fleet size for service, not counting the vehicles that are idling, cruising, or recharging. The access link connected to 105 emerges as a critical bottleneck, operating at full saturation (100\%) during periods 1 and 2, and remaining highly utilized (86\%) in period 3. Similarly, access links connected to node 106 and node 111 experience high demands, reaching 100\% utilization in periods 1 and 3, respectively. The access link connected to node 102 shows a gradual load increase, starting at moderate levels and reaching full capacity by the final period. In contrast, Node 110 is heavily underused, suggesting a potential service redundancy at this location.

\begin{figure}[htbp]
    \centering
    \includegraphics[width=0.6\textwidth]{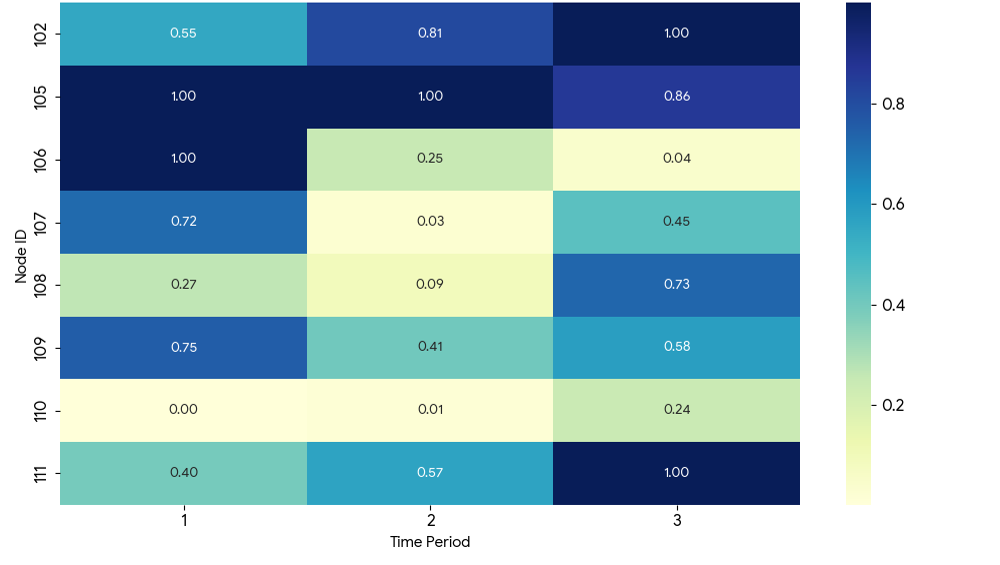}
    \caption{Heatmap of baseline case mobility service access link utilization (labeled by connected service node ID) across time periods.}
    \label{fig:service_util_heatmap}
\end{figure}

The charging infrastructure maintains high utilization without reaching saturation points. As shown in \textbf{Fig.~\ref{fig:charge_util_bar}}, station 209 consistently experiences a higher load than station 207 due to its lower charging price. It peaks at 84\% utilization in period 2, while station 207 reaches a maximum of 77\% in the same interval. Both stations maintain activity levels above 60\% across all periods, indicating effective load balancing and a steady need for fleet recharging. 

\begin{figure}[htbp]
    \centering
    \includegraphics[width=0.6\textwidth]{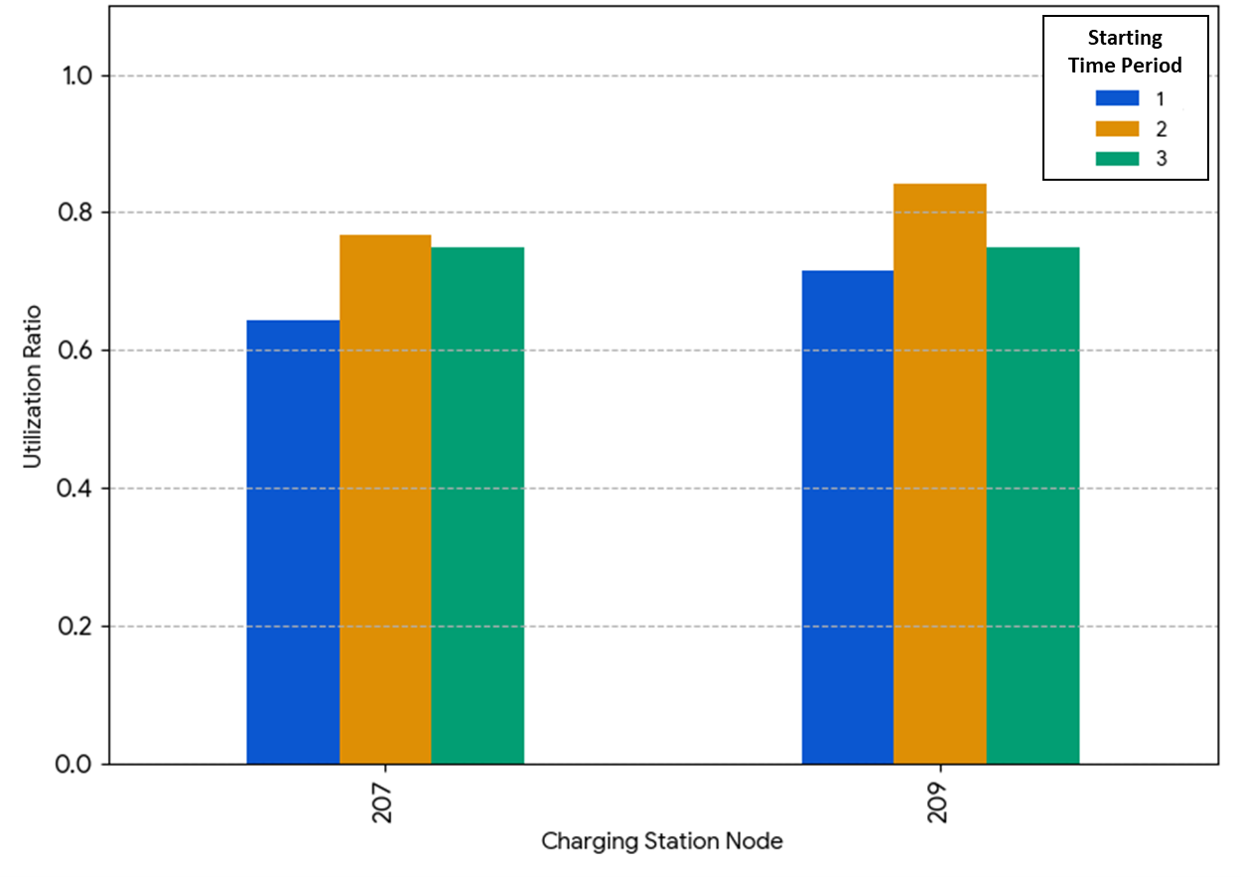}
    \caption{Utilization Levels of Charging Stations 207 and 209 across Charging Intervals.}
    \label{fig:charge_util_bar}
\end{figure}

\subsubsection{Charging station capacity sensitivity}

We conduct a sensitivity analysis on charging station capacity. Holding all other parameters constant, we vary the capacity of both charging stations from 300 to 100 units in decrements of 10. As a result, we obtain 20 results from the created instances.

As illustrated in \textbf{Fig.~\ref{fig:obj_fleet}}, system performance shows a threshold at a charging capacity of 280 units per station. Above this threshold, the system operates stably with the objective value remaining flat around 1,463 and the the active fleet size stays at 1,474 vehicles.This indicates that 280 units per charging station adequately support the fleet, and further capacity increases provide marginal benefit. Below 280 units, charging constraints restrict the active fleet, which decreases linearly to 600 vehicles as capacity drops to 100. Consequently, the objective value increases due to higher system disutility from unmet MOD demand and inefficient routing in OOP subnetwork.

\begin{figure}[htbp]
    \centering\includegraphics[width=0.8\textwidth]{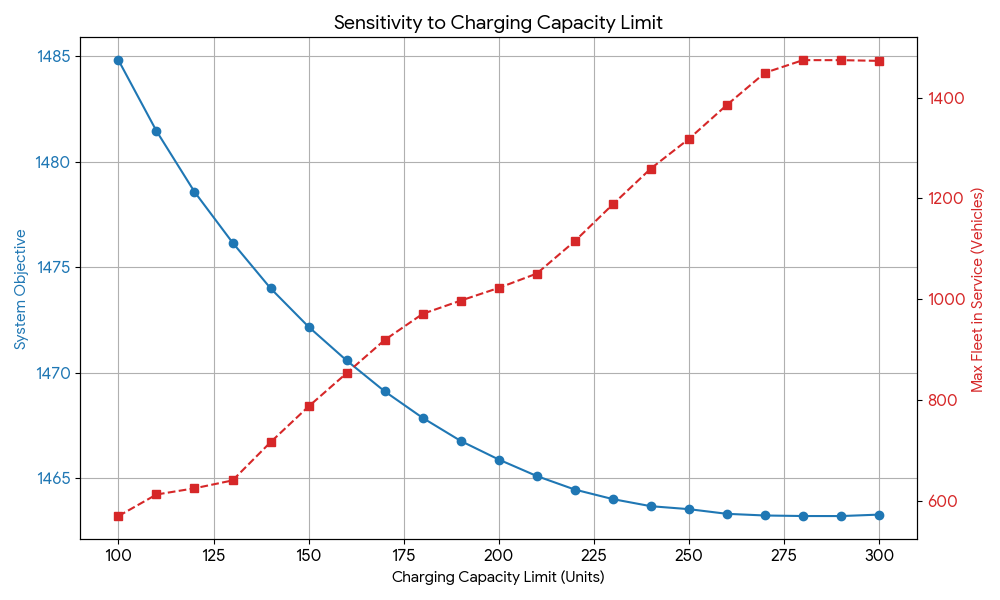}
    \caption{Sensitivity of System Objective and Maximum Fleet in Service to Charging Capacity.} \label{fig:obj_fleet}
\end{figure}

Charging infrastructure constraints directly affect MOD service availability. The heatmaps in \textbf{Fig.~\ref{fig:charging_cap_sen_z_util}} shows access link utilization across different periods. For visual clarity, access links are labeled according to their corresponding service nodes. In period 1 (\textbf{Fig.~\ref{fig:charging_cap_sen_z_util}(a)}), utilization across most access links declines as charging capacity decreases, consistent with the system-wide fleet reduction. However, access links connected to nodes 105 and 106 maintain high utilization, indicating their importance as key service locations. This is expected due to their proximity to origin 1 and 4, which generate significantly higher outbound flow than origins 2 and 3.

\begin{figure}[htbp]
    \centering\includegraphics[width=0.75\textwidth]{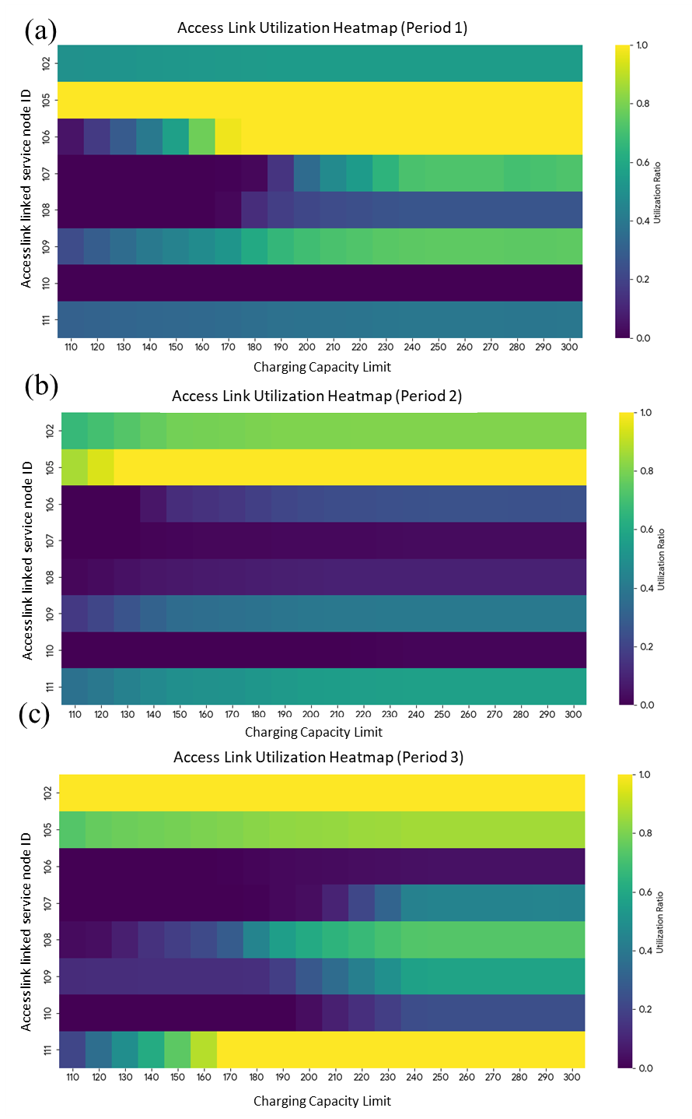}
    \caption{Access link utilization heatmap.} \label{fig:charging_cap_sen_z_util}
\end{figure}

In period 2 (\textbf{Fig.~\ref{fig:charging_cap_sen_z_util}(b)}), overall demand is lower. Access link connected to node 105 remains fully utilized despite capacity reductions, while utilization at other access links drops sharply. This shows that under vehicle shortages, the system prioritizes allocating the limited fleet to the highest-demand locations. In period 3 (\textbf{Fig.~\ref{fig:charging_cap_sen_z_util}(c)}), utilization at access links connected to nodes 106 and 110 drops to near zero as charging capacity declines. The limited fleet consolidates to serve specific nodes, such as 102 and 111. This highlights how capacity constraints force the system to sacrifice broad network coverage to maintain service on high-yield locations.

\subsubsection{Route dispersion effect sensitivity}

We analyze the model's sensitivity to the service dispersion weight, $w_{sd}$. Keeping other parameters constant (including a charging station capacity of 300 units), we vary $w_{sd}$ from 0.5 to 2.0 in 0.1 increments. We obtain a total of 16 instances.

The analysis demonstrates a linear relationship between network dispersion and system efficiency. As shown in \textbf{Fig.~\ref{fig:disp_obj_fleet}}, the objective value increases linearly from 1,301 at $w_{sd} = 0.5$ to 1,758 at $w_{sd} = 2.0$. This occurs because a higher dispersion weight penalizes route concentration, forcing passenger and vehicle flows onto more diverse, less cost-effective routes, thereby increasing total system disutility.

The active fleet size remains relatively stable, ranging from 1,453 to 1,500 vehicles, exhibiting a slight upward trend. Because users choose more dispersed routes rather than aggregating on more efficient ones, a slightly larger fleet is required to serve the network when the value of $w_{sd}$ is higher. Notably, this increase in active fleet size more closely approximates a step-wise pattern rather than a strictly continuous one. Consequently, when calibrating this parameter with empirical data, interval fitting is sufficient and computationally more efficient than precise point estimation.

\begin{figure}[htbp]
    \centering
    \includegraphics[width=0.8\textwidth]{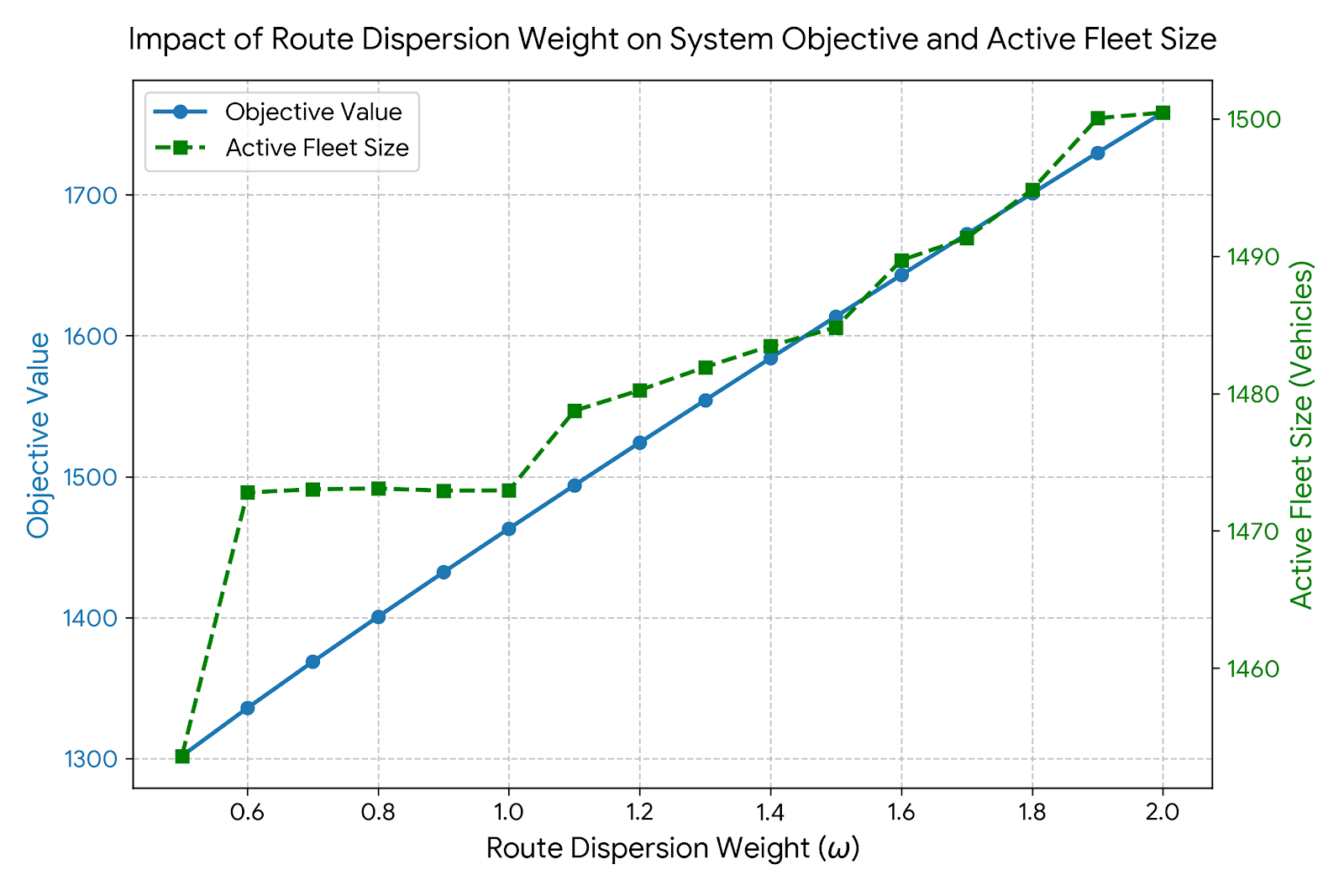}
    \caption{Sensitivity of system objective and fleet size to mobility service dispersion weight.}
    \label{fig:disp_obj_fleet}
\end{figure}

The spatiotemporal utilization of service nodes (\textbf{Fig.~\ref{fig:dispersion_sen_z_util}}) further illustrates how increasing $w_{sd}$ forces a broader spatial distribution of capacity allocation. In period 1 (\textbf{Fig.~\ref{fig:dispersion_sen_z_util}(a)}), high-demand access links like the ones linking node 105 and 106 remain fully saturated (100\%) across all weights. However, utilization of other access links shift significantly: as dispersion increases, utilization at access links connected to node 107 drops from 85\% to 55\%, while traffic is diverted to alternative service locations like node 109 (increasing from 55\% to 87\%) and node 110 (rising from near-zero to active utilization).

\begin{figure}[htbp]
    \centering
    \includegraphics[width=0.8\textwidth]{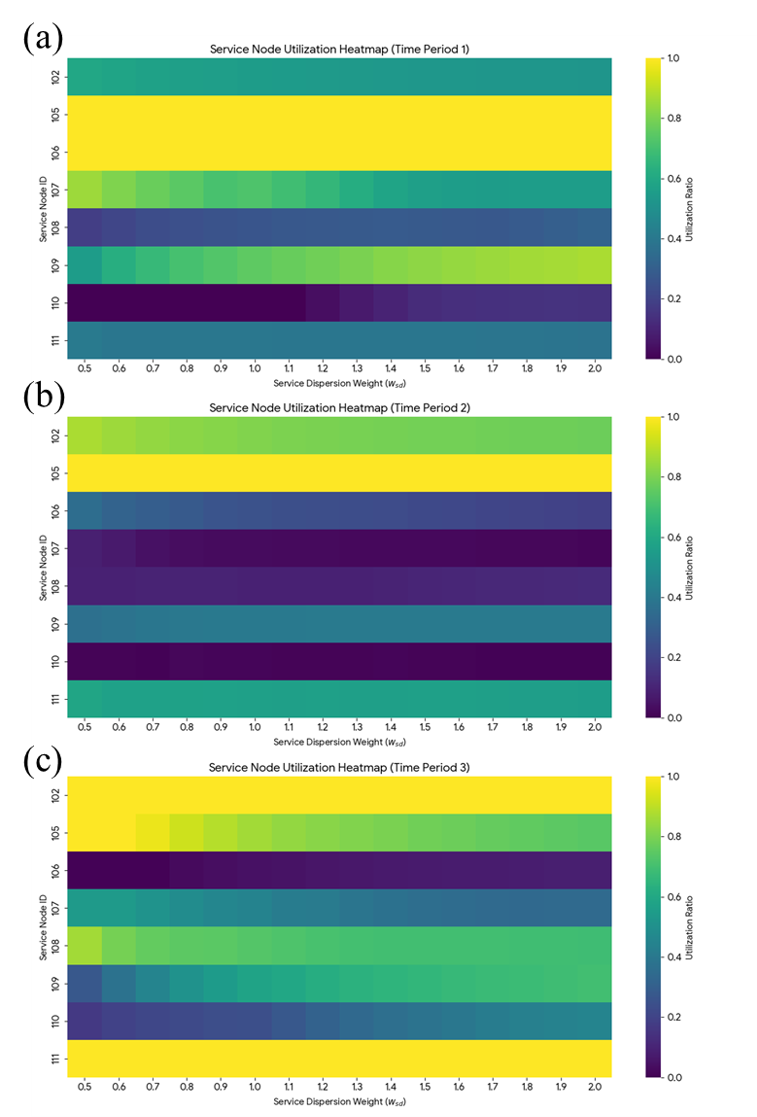}
    \caption{Service node utilization heatmap when (a) t=1, (b) t=2, and (c) t=3.}
    \label{fig:dispersion_sen_z_util}
\end{figure}

In period 2 (\textbf{Fig.~\ref{fig:dispersion_sen_z_util}(b)}), access link connected to node 105 remains highly utilized, while usage at access link connected to node 102 drops slightly from 87\% to 77\%. In period 3 (\textbf{Fig.~\ref{fig:dispersion_sen_z_util}(c)}), the dispersion effect of $w_{sd}$ is most prominent. At low dispersion weights ($w_{sd} \le 0.8$), access link to node 105 is fully utilized (100\%). As $w_{sd}$ approaches 2.0, the system shifts traffic toward underutilized locations. Consequently, the utilization of the access link to node 105 drops to 73\%, while the access link to node 109 increases from 27\% to 69\%, and the value for the access link entering node 110 rises from 16\% to 45\%.

Ultimately, the service dispersion weight serves as a key empirical parameter capturing the diversity of route choices within the network. In practice, $w_{sd}$ must be estimated directly from real-world travel data. For networks with highly heterogeneous routing behaviors, higher values shall be fitted to $w_{sd}$.

\subsubsection{Bilevel optimization: Sensitivity to fleet size constraints with service-wide distance-based pricing}

We focus on a long-term resource strategy by evaluating the system's sensitivity to varying total fleet size and their endogenously linked operational costs. By optimizing both the upper-level pricing decisions and lower level allocation responses, we explore how the platform and operators jointly optimize dynamic service pricing and fleet deployment to maximize profitability. Throughout this analysis, we employ a platform-wide, time-dependent distance pricing strategy, applying a single uniform rate across the entire MOD subnetwork for each of the three time intervals.

The fleet size parameter $V_m$ controls the total number of vehicles available in the network. The scale of this fleet directly impacts and operator cost coefficients. Specifically, an expanded fleet incurs higher unit operating costs due to increased vehicle idling and deadhead cruising. Given a known initial fleet size $V_m$ and operating costs parameters ($c_l$ and $c_i$), we assume a relationship between them for scaling cost parameters to different fleet size scenarios for input into the model. The relationships are formulated in \textbf{Eq.~\ref{eq:fleet_cost}} and \textbf{Eq.~\ref{eq:node_cost}}:

\begin{equation}\label{eq:fleet_cost}
    c_l = \underline{c}_{l} \times (V_m/\underline{V})^{\lambda}, \quad \forall m\in M, l\in A_m
\end{equation}

\begin{equation}\label{eq:node_cost}
    c_i = \underline{c}_{i} \times (V_m/\underline{V})^{\lambda}, \quad \forall m\in M, i\in N_m
\end{equation}

These relationships follow a Cobb-Douglas functional form. Here, $\underline{c}_l$ and $\underline{c}_i$ represent the baseline link travel time, link operating cost, and node operating cost, respectively, calibrated at a baseline fleet size of $\underline{V}$. For this sensitivity analysis, $\underline{t}_l$ corresponds to the baseline link times defined in the previous section. We define the baseline cost parameters as $\underline{c}_l = 0.2$ and $\underline{c}_i = 1$ at a reference fleet size of $\underline{V} = 500$. The parameter $\lambda$ represents the elasticity of the cost function with respect to fleet size. We follow the empirical findings for modern e-hailing systems in \cite{zhang2019efficiency}, noting that microtransit would tend to be less elastic than ridehail, and set $\lambda = 0.2$ accordingly.

We solve the bi-level model by varying the total fleet size from 500 to 1,600 vehicles in increments of 50, keeping all other parameters constant except MOD service pricing. This illustrates how fleet scaling influences optimal service pricing and resource allocation. For each scenario, we update the relevant cost coefficients and optimize the model using the framework of \textbf{Algorithm \ref{alg:pricing_optimization_decomposition}}. Given the moderate scale of this case study, we bypass the proposed lower-level heuristic and directly solve the integrated lower-level model using Gurobi while preserving the remaining algorithmic steps. Executed with 15 multi-start points and a maximum of 15 iterations, each scenario requires approximately 25 minutes of runtime.

To reflect the platform's economics, we introduce a fleet deployment cost, $c_m \times V_m$, to the upper-level objective function. The unit deployment cost, $c_m$, incorporates economies of scale by scaling with fleet size, defined by a baseline cost of $\underline{c}_m = 1$ (at a reference fleet $\underline{V} = 500$) and an elasticity parameter of $\lambda = -0.2$. This formulation explicitly captures the trade-off between fleet capital investment and the operational penalties of supply-demand mismatches.

The final results meet the budget constraints requirement. \textbf{Fig.~\ref{fig:profit_analysis}} shows that operational and total profit exhibit concave growth relative to total fleet size. \textbf{Fig.~\ref{fig:fleet_utilization}} illustrates the corresponding supply-demand dynamics by comparing total fleet size with active fleet size. The gap between total and active fleet size widens significantly above 800 vehicles. Below this threshold, over 93\% of the fleet operates during peak periods, indicating minimal fleet idling due to charging. Because passenger demand is fully met at 750 vehicles, adding vehicles beyond this level causes oversupply and sharply reduces fleet utilization.

In supply-constrained scenarios (capacity below 800), each additional vehicle captures unmet demand and increases revenue. Beyond 1,000 vehicles, the marginal benefit of adding vehicles diminishes. Platform profit plateaus between 650 and 750 vehicles, which remains entirely within the supply-constrained scenario. Total profit peaks at approximately \$55,800 before declining steadily as fleet size increases. This demonstrates that with endogenous operating costs, the platform maximizes profitability before market supply saturates. Because profit peaks across a fleet size range rather than at a single point, operators have strategic flexibility in fleet deployment, provided they avoid oversupply.

\begin{remark} \label{remark: zero_subsidy}
\textit{There exists an optimal fleet size for mobility providers to maximize profit in an mobility-energy market when considering operating costs as a function of fleet size, serving as a critical threshold between under-supplied and over-supplied scenarios in fleet size.}
\end{remark}

\begin{figure}[htbp]
    \centering
    \includegraphics[width=0.6\textwidth]{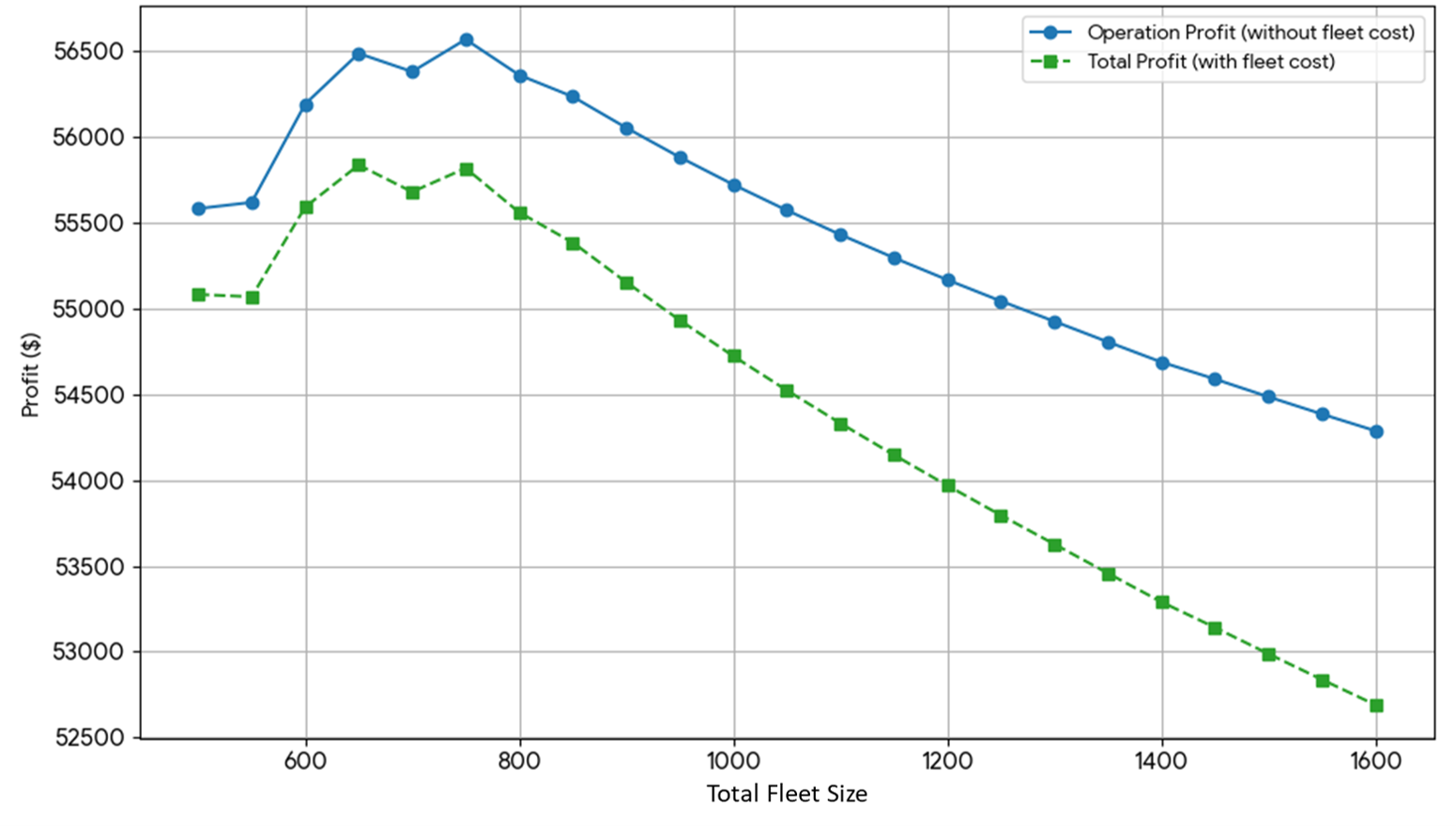}
    \caption{Comparison of Operator Profit and True Profit under varying Fleet Size.}
    \label{fig:profit_analysis}
\end{figure}

\begin{figure}[htbp]
    \centering
    \includegraphics[width=0.6\textwidth]{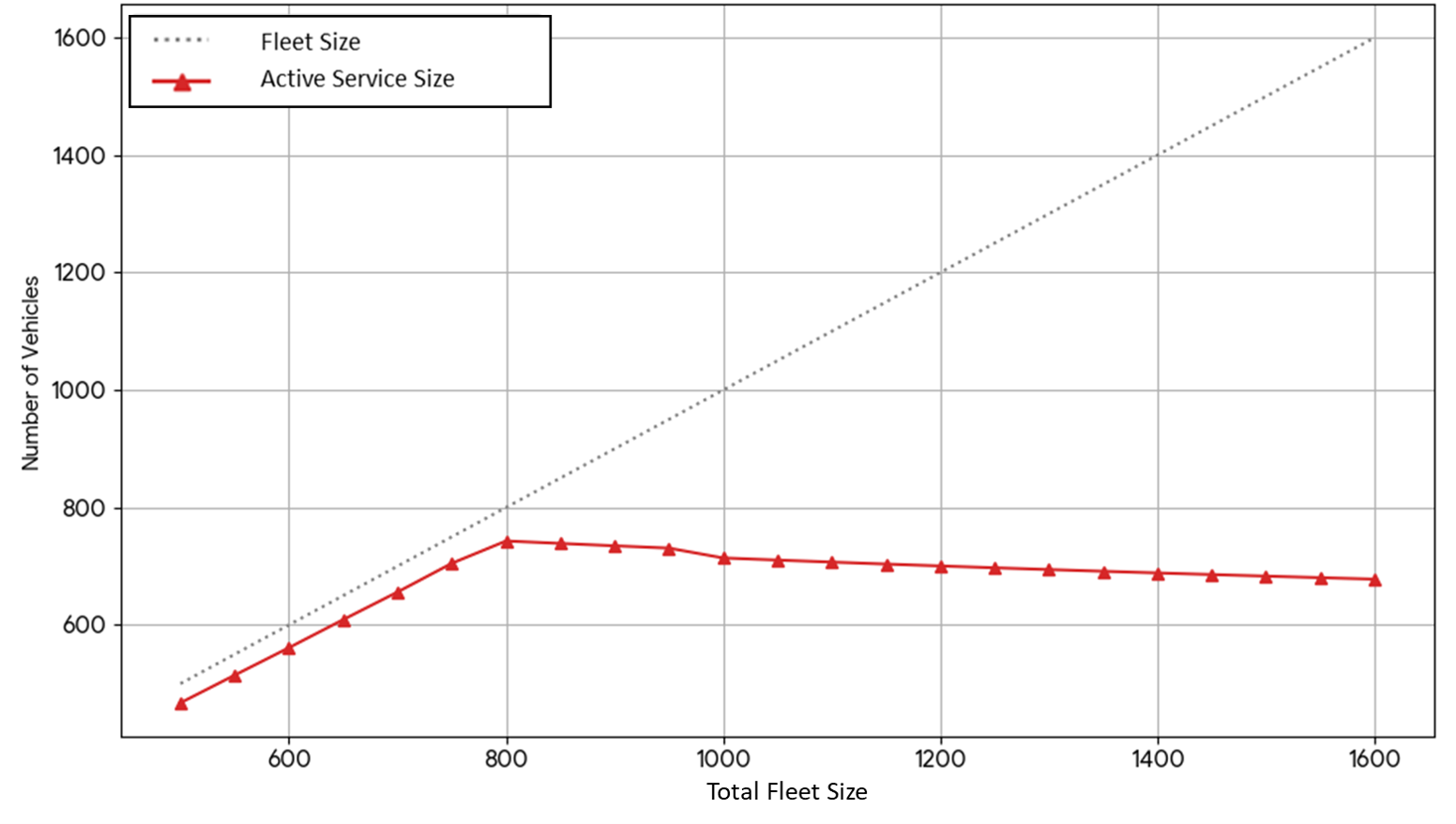}
    \caption{Fleet Utilization Dynamics: Fleet Size vs. Active Service Size.}
    \label{fig:fleet_utilization}
\end{figure}

\textbf{Fig.~\ref{fig:pricing_strategy}} illustrates the optimal time-dependent pricing strategy across varying fleet capacities. This dynamic pricing mechanism manages the transition from supply scarcity to oversupply. Under severe constraints (e.g., 500 vehicles), prices peak to ration the limited fleet. As fleet size increases, the platform lowers peak-hour prices to ensure new vehicles are utilized, thereby increasing total profit. Upon reaching market saturation, the pricing structure stabilizes.

\begin{remark} \label{remark: pricing_pattern}
\textit{The optimal pricing trends across time periods may fluctuate under under-supplied fleet sizes and exhibit steady trends under over-supplied fleet sizes.}
\end{remark}

This insight suggests that for steady pricing planning, it is more reliable to risk over-supply than under-supply in fleet size.

\begin{figure}[htbp]
    \centering
    \includegraphics[width=0.6\textwidth]{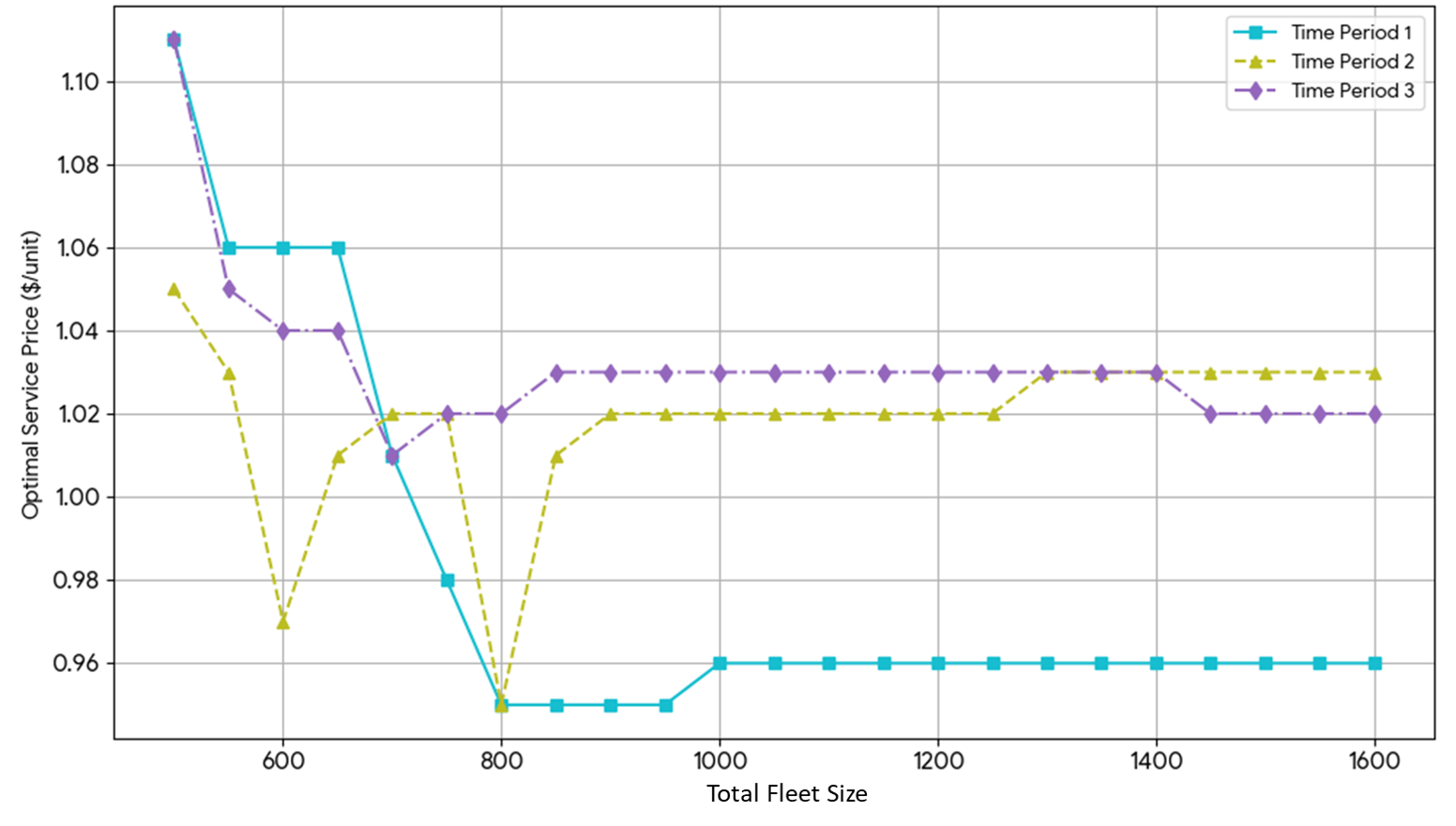}
    \caption{Optimal Dynamic Pricing Strategy per Time Period.}
    \label{fig:pricing_strategy}
\end{figure}

Peak-hour prices (periods 1 and 3) decrease nearly monotonically before stabilizing in the oversupply regime. However, off-peak pricing (period 2) fluctuates before reaching stability, indicating high sensitivity to fleet transitions. During period 2, pricing functions as a secondary buffer to manage residual supply, network rebalancing, and charging demand, rather than driving primary profit. Thus, the system prioritizes peak-period demand capture while leveraging off-peak pricing to maintain operational equilibrium. Despite these dynamic adjustments, optimal prices remain confined between \$0.95 and \$1.11, representing a variance of less than 17\%. This narrow range indicates high price elasticity of demand; aggressive price surging would significantly reduce ridership and revenue. Consequently, maximizing vehicle utilization and trip volume is more critical to profitability than aggressive pricing.

\begin{remark} \label{remark: pricing_dynamics}
\textit{With under-supplied fleet sizes, prices in high-demand periods follow a steady trajectory while low-demand prices fluctuate to accommodate residual supply and system transition.}
\end{remark}

This result suggests that more attention needs to be paid to the optimal pricing in low-demand periods because they are less predictable.

\subsection{Bilevel optimization: heterogeneous MOD fleets}

We expand the scenario to convert the single MOD fleet into two heterogeneous fleets participating in the platform. All final results meet the budget constraints requirement. As illustrated in \textbf{Fig.~\ref{fig:service_network_duo}}, the fleets cover distinct regions with an overlapping service area. Operator 1 (blue region) manages a fleet of 700 vehicles across 6 service nodes, while Operator 2 (orange region) operates 300 vehicles across 4 service nodes. Their service areas overlap at nodes 6 and 9 in the base OOP network. Both fleets adhere to the cost scaling rules defined in \textbf{Eq.~\ref{eq:fleet_cost}} and \textbf{Eq.~\ref{eq:node_cost}}, applying the same elasticity ($\lambda = 0.2$) and reference fleet size ($\underline{V} = 500$). The link attributes for both MOD subnetworks and the passenger demand patterns remain identical to the baseline case.

For the subsequent analyzes, we evaluate two dynamic pricing strategies. The first is an operator-specific distance-based pricing strategy, where each operator independently optimizes its fee in each time period. The second is a platform-wide distance-based strategy, which enforces a single, unified per-distance fee across both fleets in each period serving as a bundled fare.

\begin{figure}[htbp]
    \centering
    \includegraphics[width=0.6\textwidth]{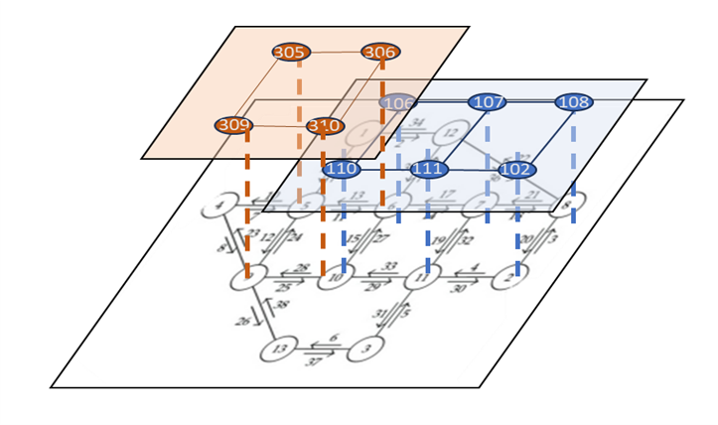}
    \caption{Service network with two MOD service providers.}
    \label{fig:service_network_duo}
\end{figure}

\textbf{Table~\ref{tab:pricing_comparison}} summarizes the two pricing strategies used in the two-fleet scenario and the results from the previous case study with a total fleet size of 1000. The single-fleet case yields the highest profit compared to both profits under the two-fleet scenarios Relative to the single-fleet benchmark, operator-specific pricing with a two-fleet configuration decreases profit by 4.16\%, and profit further decreases by 5.14\% when platform-wide pricing is enforced.

In the two-fleet setting, operator-wide pricing strictly dominates platform-wide pricing from a revenue maximization perspective. Profit increases from \$52,857.06 to \$53,404.41, while the maximum active fleet size falls from 898 to 708, a reduction of 190 vehicles (21.2\%). When the pricing rule varies across operators, the system can internalize differences in effective capacity allocations at the operator-level. The optimal price profiles are consistent with this interpretation. Under operator-wide pricing, operator 2 charges more than operator 1 in all three periods. This pattern indicates that the optimal pricing rule differentiates systematically across operators. In contrast, the platform-wide rule compresses cross-operator variation into a single average price. 

When vehicles are pooled, capacity can be allocated more effectively. The optimal prices under the single-fleet scenario are lower than those under the two-fleet scenarios while obtaining higher total profit. This is mainly due to the reduced capacity allocation cost. Under the two-fleet scenarios, travelers need to switch between operators when making their trips, leading to additional capacity required at the overlapping area. Pooling reduces the inefficiencies created by operator boundaries and allows the system to attain a higher-profit allocation with only a modest peak fleet requirement.

\begin{table}[htbp]
\centering
\caption{Comparison of outcomes under different fleet and pricing structures}
\label{tab:pricing_comparison}
\renewcommand{\arraystretch}{1.15}
\small
\begin{tabularx}{\linewidth}{ll>{\centering\arraybackslash}Xcc}
\toprule
& & \multicolumn{2}{c}{Two fleets} & Single fleet \\
\cmidrule(lr){3-4}
Metric & Detail & \makecell{Operator-Specific\\pricing} & \makecell{Platform-wide\\pricing} & Pricing \\
\midrule
Profit (\$)& -- & 53404.41 & 52857.06 & 55721.60 \\

\multirow{3}{*}{\makecell{Price (\$)}}
& Period 1 & Operator 1: 1.07; Operator 2: 1.28 & 1.10 & 0.96 \\
& Period 2 & Operator 1: 1.03; Operator 2: 1.27 & 1.08 & 1.02 \\
& Period 3 & Operator 1: 1.06; Operator 2: 1.34 & 1.03 & 1.03 \\

\makecell{Maximum active\\fleet size} & -- & 708 & 898 & 714 \\
\bottomrule
\end{tabularx}
\end{table}

\begin{remark} \label{remark: twofleets}
\textit{The proposed model captures the inefficiency of mobility-energy markets with multiple fleet operators and corresponding operating costs of each individual fleet. When bundling fares, the combined fleet size increases to accommodate the common fare used by the platform.}
\end{remark}

\section{Discussion and Conclusion}

The transition to mobility-energy markets shifts urban transportation from a two-sided market to a three-sided ecosystem comprising travelers, mobility providers, and energy providers. This paper proposes a bilevel assignment game to count the additional interdependencies between mobility demand and energy infrastructure. Using a link-based PURC framework, the model achieves high computational efficiency and effectively models the temporal lag of charging demand, enabling scalable, multi-period optimization of fleet distribution, capacity allocation, and dynamic pricing. Using an expanded Nguyen--Dupuis network, the case study shows why mobility operations and energy constraints should be modeled together rather than as separate planning layers: as service demand varies across time, charging demand propagates across periods, and the platform’s strategic decisions must balance traveler utility, operator cost, and infrastructure capacity simultaneously. The case study demonstrates the practical value of the proposed framework as both an operational analysis tool and a long-term planning instrument.

At the operational level, the model identifies clear bottlenecks in specific MOD access links while pinpointing underused locations. At the same time, both charging stations remain heavily used without becoming fully saturated, with the cheaper station attracting consistently higher load. These patterns show that even when total fleet can meet the required demand, performance is still constrained by localized service-node pressure and the spatial distribution of charging demand. In other words, fleet adequacy alone does not guarantee operational efficiency; access capacity and charging resource differences also shape equilibrium outcomes.

At the strategic level, the model shows that profitability is maximized not by only expanding fleet size, but by sizing the fleet near the upper bound of the supply-constrained scenario. As fleet size grows, additional vehicles initially improve demand capture and raise revenue, but the benefit diminishes once demand is fully served. Oversupply widens the gap between total and active fleet, increases idle and repositioning burdens, and reduces utilization, causing total profit to flatten and then decline. Therefore, platform planners should target high service utilization before saturation rather than pursue fleet expansion beyond the point of effective demand absorption.

Optimal prices vary with fleet availability and demand conditions, but remain within a relatively narrow band, indicating that the system is sensitive to price increases and that sharp surges would likely suppress ridership more than they would improve profitability. Prices during high-demand periods decline steadily as the fleet expands, while low-demand prices fluctuate more in response to transitional conditions, residual supply, and charging needs. This implies that in integrated systems, pricing should be interpreted as a tool for coordinating utilization across time. In practice, platform performance depends more on maintaining vehicle productivity and trip volume than on imposing high markups.

Methodologically, the case study also supports the computational efficiency of the framework. For the moderate-size instance, the integrated lower-level problem can be solved directly to global optimality using commercial solvers, while the proposed decomposition and AM heuristic achieves a near-optimal solution with negligible error and substantially shorter runtime. This suggests that the proposed heuristic has strong potential for larger instances.

Several limitations should be acknowledged. First, the current model does not capture non-linear service-charging relationships. Second, the framework does not include the dynamic part of the power grid, such as localized voltage stability and time based electricity pricing fluctuations. Integrating these smart-grid constraints into the charging subproblem would provide a more complete view of the three-sided mobility-energy ecosystem. Furthermore, future models could incorporate heterogeneous energy providers with distinct cost structures and technological capabilities. This inclusion would reflect the growing complexity of modern charging infrastructure, accommodating varying charging speeds and emerging options such as Charging-as-a-Service (CaaS) networks with more explicit charging provider behaviors like delivery fleet allocations and pricing decisions. 

\section*{Acknowledgment}
The authors are grateful for funding support from the U.S. National Science Foundation, grant number CMMI-2423908. 

\bibliographystyle{apalike}
\bibliography{reference}

\clearpage
\section*{Appendix A}
\setcounter{table}{0}
\renewcommand{\thetable}{A\arabic{table}}

\setlength{\LTpre}{6pt}
\setlength{\LTpost}{6pt}

{\renewcommand{\arraystretch}{1.2}

\begin{longtable}{|c|c|c|r|}
\caption{Passenger Demand ($q_{s}^{t}$) per OD Pair and Time Period}
\label{tab:demand_data}\\
\hline
\textbf{Time ($t$)} & \textbf{Origin ($O$)} & \textbf{Destination ($D$)} & \textbf{Demand} \\ \hline
\endfirsthead

\hline
\textbf{Time ($t$)} & \textbf{Origin ($O$)} & \textbf{Destination ($D$)} & \textbf{Demand} \\ \hline
\endhead

\hline
\multicolumn{4}{r}{Continued on next page} \\
\endfoot

\hline
\endlastfoot

1 & 1 & 2 & 600 \\ \hline
1 & 1 & 3 & 1200 \\ \hline
1 & 4 & 2 & 900 \\ \hline
1 & 4 & 3 & 300 \\ \hline
1 & 2 & 1 & 120 \\ \hline
1 & 3 & 1 & 240 \\ \hline
1 & 2 & 4 & 180 \\ \hline
1 & 3 & 4 & 60 \\ \hline

2 & 1 & 2 & 180 \\ \hline
2 & 1 & 3 & 360 \\ \hline
2 & 4 & 2 & 270 \\ \hline
2 & 4 & 3 & 90 \\ \hline
2 & 2 & 1 & 180 \\ \hline
2 & 3 & 1 & 360 \\ \hline
2 & 2 & 4 & 270 \\ \hline
2 & 3 & 4 & 90 \\ \hline

3 & 1 & 2 & 120 \\ \hline
3 & 1 & 3 & 240 \\ \hline
3 & 4 & 2 & 180 \\ \hline
3 & 4 & 3 & 60 \\ \hline
3 & 2 & 1 & 600 \\ \hline
3 & 3 & 1 & 1200 \\ \hline
3 & 2 & 4 & 900 \\ \hline
3 & 3 & 4 & 300 \\ \hline

\end{longtable}
}

{\small
\renewcommand{\arraystretch}{1.2}

\begin{xltabular}{\textwidth}{|>{\raggedright\arraybackslash}X|c|}
\caption{Sets, Parameters, and Variables of the Integrated eMaaS Model}
\label{tab:np_input}\\
\hline
\textbf{Input element} & \textbf{Value} \\ \hline
\endfirsthead

\hline
\textbf{Input element} & \textbf{Value} \\ \hline
\endhead

\hline
\multicolumn{2}{r}{Continued on next page} \\
\endfoot

\hline
\endlastfoot

Service node capacity limit & 300 \\ \hline
Charging capacity limit & 300 \\ \hline
Total fleet size & 1600 \\ \hline
Operational cost per unit length on MOD link & \$0.2/mile \\ \hline
Travel time on base link & 15 miles/h \\ \hline
Travel time on MOD link & 25 miles/h \\ \hline
Traveler value of time (VOT) & \$20/h \\ \hline
Service price per unit length on MOD link & \$0.5/mile \\ \hline
Fixed link price per link usage in OOP network & \$1 \\ \hline
Operator cost per unit length on recharge network link & \$0.1/mile \\ \hline
Recharge price per usage on node 207 & \$4 \\ \hline
Recharge price per usage on node 209 & \$2 \\ \hline
Unavailable fleet buffer ratio & 0.2 \\ \hline
Temporal propagation factor from service demand to charging load & [0.05, 0.1, 0.2] \\ \hline
Overall weight between service and recharge disutility & 0.01 \\ \hline
Service dispersion weight & 1 \\ \hline
Traveler utility weight & 1 \\ \hline
Operator utility weight & $5\times 10^{-4}$ \\ \hline
Recharge dispersion weight & 2 \\ \hline
Operator recharge utility weight & 1 \\ \hline
Recharge station utility weight & $10^{-3}$ \\ \hline

\end{xltabular}
}

\end{document}